%% file: mmrp.tex
\documentclass{llncs}

\usepackage{cite,caption,subcaption}

\usepackage{amsmath,enumerate,amssymb,amsfonts,xcolor,algorithm2e,url}
\usepackage{subfig}
\usepackage{graphicx,cleveref}           
\usepackage{myDefs}
\newcommand*\samethanks[1][\value{footnote}]{\footnotemark[#1]}
\crefname{equation}{Constraint}{Constraints}
\crefname{theorem}{Theorem}{Theorems}
\crefname{definition}{Definition}{Constraints}
\crefname{figure}{Figure}{Figures}
\crefname{section}{Section}{Sections}
\crefname{lemma}{Lemma}{Lemmas}
\crefname{proposition}{Proposition}{Propositions}
\crefname{appendix}{Appendix}{Appendixes}
\crefname{nclaim}{Claim}{Claims}

\def \fw {4 \cdot \omega}
\def \hbp {\hat{B'}}
\newcommand{\floor}[1]{\lfloor #1 \rfloor}

\newcommand{\mra}{\mbox{\textsc{Sum-MRA}}}
\newcommand{\smr}{\mbox{\textsc{Sum-MR}}}
\newcommand{\stp}[2]{\mbox{Append-path($#1$, $#2$)}}
\newcommand{\stpn}{\mbox{Append-path}}

\title{Approximation Algorithms for Movement Repairmen}
\author{
MohammadTaghi Hajiaghayi\inst{1}\thanks{Supported in part by NSF CAREER award 1053605, NSF grant CCF-1161626,
      ONR YIP award N000141110662, DARPA/AFOSR grant FA9550-12-1-0423,
      and a University of Maryland Research and Scholarship Award (RASA).}
\and Rohit Khandekar\inst{2} \and M. Reza Khani\inst{3}\samethanks \and
Guy Kortsarz \inst{4} \thanks{Supported in part by NSF award number 434923}
}
\institute{University of Maryland,
      College Park, MD; and AT\&T Labs,
      Florham Park, NJ. \protect\url{hajiagha@cs.umd.edu}.
\and Knight Capital Group, Jersey city, NJ. rkhandekar@gmail.com.
\and University of Maryland, College Park, MD. khani@cs.umd.edu.
\and Rutgers University, Camden, NJ. guyk@camden.rutgers.edu.
}

\date{\today}

\begin{document}
\maketitle
\begin{abstract}
In the {\em Movement Repairmen (MR)} problem we are given a metric
space $(V, d)$ along with a set $R$ of $k$ repairmen $r_1, r_2,
\ldots, r_k$ with their start depots $s_1, s_2, \ldots, s_k \in V$ and
speeds $v_1, v_2, \ldots, v_k \geq 0$ respectively and a set $C$ of
$m$ clients $c_1, c_2, \ldots, c_m$ having start locations $s'_1,
s'_2, \ldots, s'_m \in V$ and speeds $v'_1, v'_2, \ldots, v'_m \geq 0$
respectively. If $t$ is the earliest time a client $c_j$ is collocated
with any repairman (say, $r_i$) at a node $u$, we say that the client
is served by $r_i$ at $u$ and that its latency is $t$. The objective
in the (\smr{}) problem is to plan the movements for all repairmen and
clients to minimize the sum (average) of the clients latencies.  The
motivation for this problem comes, for example, from Amazon Locker
Delivery \cite{amazon} and USPS gopost \cite{gopost}.  We give the first $O(\log
n)$-approximation algorithm for the \smr{} problem.
In order to solve \smr{} we formulate an LP for the problem
and bound its integrality gap.  Our LP has exponentially many
variables, therefore we need a separation oracle for the dual LP.
This separation oracle is an instance of Neighborhood Prize
Collecting Steiner Tree (NPCST) problem in which we want to find a tree with
weight at most $L$ collecting the maximum profit from the clients by
visiting at least one node from their neighborhoods. The NPCST
problem, even with the possibility to violate both the tree weight and
neighborhood radii, is still very hard to approximate.  We deal with
this difficulty by using LP with geometrically increasing segments of
the time line, and by giving a {\em tricriteria} approximation for the
problem.  The rounding needs a relatively involved analysis.  We give
a constant approximation algorithm for \smr{} in Euclidean Space where
the speed of the clients differ by a constant factor.  We also give a
constant approximation for the makespan variant.
\end{abstract}

\section{Introduction}
In the well-known {\em Traveling Repairman (TR)} problem,
the goal is to find a tour to cover a set of clients such that the sum
of latencies seen by the clients is minimized. The problem is also known as the {\em minimum latency problem}, see \cite{GK98},
the {\em School-bus driver problem}, see \cite{W93} etc.
This problem is well studied in the operations research literature and has lots of applications in real world, see for example~\cite{BCR93}. The problem is NP-Hard even in tree metrics \cite{S06}. Blum et al. \cite{BCC+94} give the first constant-factor approximation algorithm for the TR problem. They also observe that there is no PTAS ($(1+\epsilon)$-approximation algorithm for an arbitrary constant $\epsilon > 0$) for the problem unless $P = NP$. After a sequence of improvements,
Chaudhuri et al. \cite{CGRT03} give a $3.59$-approximation algorithm for TR which is the current best approximation factor for this problem.

Fakcharoenphol et al. \cite{FHR03} generalize the TR problem to the $k$-Traveling Repairman ($k$-TR) problem in which instead of one repairman, we can
use $k$ repairmen to service the clients where all the repairmen start
from the same depot. They give a $16.994$-approximation algorithm for $k$-TR.
Chekuri and Kumar \cite{CK04} give a $24$-approximation algorithm for the ``multiple-depot'' version of the $k$-TR problem where the repairmen can start from different depots.

Chakrabarty and Swamy \cite{CS10} give a constant-factor approximation algorithm for the classical TR problem by introducing two new LPs. Their work is significant as it is the first LP approach to solve the problem. 

We generalize \cite{CS10} for the $k$-TR problem by allowing the repairmen
to start from different starting depots and to have different speeds.
More importantly, we give the clients ability to move with different speeds, which makes the problem significantly harder.
We formally define \smr{} as follows.
\begin{definition}
\label{def-problem}
In the \smr{} problem the inputs are given as follow.
\begin{itemize}
\item A metric space $\M = (V, d)$ where $V$ is the set of nodes and $d: (V \times V) \rightarrow \mathbb{Q}^+$ is the distance function.
\item A set $R$ of $k$ repairmen $r_1, r_2, \ldots, r_k$. Each repairman $r_i$ has a start depot $s_i \in V$ and speed $v_i \in \mathbb{Q}^+$.
\item A set $C$ of $m$ clients $c_1, c_2, \ldots, c_m$. Each client $c_j$ has a start location $s'_j \in V$ and speed $v'_j \in \mathbb{Q}^+$.
\end{itemize}
A solution to the problem consists of the following.
\begin{itemize}
\item A pair $(u_j, t_j)$ for each client $c_j$ such that $c_j$ can reach node $u_j$ by time $t_j$ considering its speed $v'_j$ (\ie $d(s'_j, u_j) \leq v'_j \cdot t_j$).
\item A path $p_i$ for each repairman $r_i$. In general $p_i$ may not be a simple path and can contain a node or an edge multiple times. Repairman $r_i$ can travel along $p_i$ with maximum speed $v_i$.
\item For each pair $(u_j, t_j)$ assigned to client $c_j$ there has to be at least one repairman ($r_i$) such that $r_i$ visits $u_j$ at time $t_j$ when it travels path $p_i$.
\end{itemize}
The objective for \smr{} is to minimize $\sum_{j = 1}^m t_j$.
\end{definition}

The problem is very natural and is also motivated by the following
real-world scenario. Amazon Locker Delivery is an optional shipping
method in Amazon online stores.  In this method clients have an option
to select a certain locker location to pick up their purchased items.
Afterwards, Amazon puts the items into a locker in the specified
location and sends the locker number and its key combination, to the
customer. The package can be picked up by the customer who can go to
the locker location by her own means. A very similar delivery option is also offered by the United States Postal Service which is known as gopost \cite{gopost}.

Our algorithm can be used directly in order to plan the movements
to minimize the average latency (or the maximum latency). Here the locations
of the Amazon stores, clients' homes and locker locations can be thought as
the nodes in the metric space in \smr{} and the repairmen are the
shipping vehicles starting from the Amazon stores with different speeds.
Moreover, we can take as input how customers are going to pick up their
packages (\eg by a car, public transport, bike, and etc.) which realizes the different speeds for the clients.   Note that unlike
\smr{} in this scenario it is not necessary for both a repairman and a
client to meet at the same node and the same time in order to serve;
but if a repairman visits a node at time $t$ a client can visit the
node at any time after $t$ and still get served.  We formalize these
methods of serving and show that the difference in the
objective of \smr{} for the two methods is at most $3 + \epsilon$ in \cref{sec:prel}.

In \cref{connections} we describe connection of our problem to the movement
framework, neighborhood TSP problems, and orienteering problems
respectively.  In \cref{sec-results} we give the outlines of our
techniques and summarize all our results. \cref{sum-sec} contains the
detailed explanation of our method to solve \smr{} in three
subsections. Subsection \ref{sec:prel} contains the necessary
preliminaries, in Subsection \ref{sec:lp} we give our LP formulation,
show how to solve it in Subsection \ref{solving-plp}, and finally in
Subsection \ref{sec:rounding} we show how to round a fractional
solution to the LP to get an integral solution to the \smr{}. It turns
out that the separation oracle for our LP is a generalization of the
Neighborhood TSP problem (to be defined in Subsection
\ref{sec-intro-ntsp}), we give the results related to the separation
oracle problem in \cref{npcst-sec}. In \cref{sec:euclidean} we give our result for the Euclidean space. \cref{max-mr} contains all the
materials related to the Max-MR problem in which instead of minimizing 
the sum of the latencies  seen by the clients we want to minimize the 
latency of the last client we visit.

\input{appendixintors.tex}

\input{results.tex}


\section{The Sum Movement Repairmen Problem}
\label{sum-sec}
\subsection{Preliminaries}
\label{sec:prel}
First we formalize the conditions that have to be met in order to serve the clients. If a client collocates with a repairman at a node $u$ of the metric space at time $t$ we say it is served {\em perfectly} with latency $t$. On the other hand, if a repairman visits a node $u$ at time $t$ and a client visits $u$ at time $t' \geq t$ we say the client is served {\em indirectly} with latency $t'$.

The following lemma shows serving indirectly instead of perfectly does not change the total latency by more than a constant factor.
\begin{lemma}
\label{serving-method}
Suppose a solution ($sol$) to \smr{} has sum of latencies $l$ where all the clients are served indirectly, then $sol$ can be transformed to a solution ($sol'$) in which all the clients are served perfectly with sum of latencies at most $(3 + \epsilon) \cdot l$ where $\epsilon > 0$ is a fixed constant.
\end{lemma}
\begin{proof}
Remember \cref{def-problem}, solution $sol$ assigns a path $p_i$ to each repairman $r_i$ and a node $u_j$ to each client $c_j$ such that $c_j$ can go to $u_j$ by time $t_j$ while a repairman has visited $u_j$ before or at time $t_j$. Suppose
repairman $r_i$ can travel $p_i$ in $t_i$ units of time considering its speed. In other words, the length of $p_i$ is at most $v_i \cdot t_i$ where $v_i$ is the speed of $r_i$. When $r_i$ serves the clients indirectly it is better for him to travel $p_i$ as fast as possible and does not wait for the clients to arrive since the clients can arrive in the nodes of $p_i$ later and are still served indirectly.

We design the movements in $sol'$ as follows. The movement for the clients are the same as $sol$, each client $c_j$ is assigned to same node $u_j$ as in $sol$ and go to the assigned node by time $t_j$. Each repairman $r_i$ starts from depot $s_i$ (its starting node) and goes one unit of time along path $p_i$ with its maximum speed $v_i$ and comes back to $s_i$, we refer to this as round $0$, then it goes $\alpha$ units of time and comes back to $s_i$ (round $1$) where $\alpha = 1 + \frac{2}{\epsilon}$. In general at each round $x$ it travels $\alpha^x$ units of time along $p_i$ and comes back. If at round $y$ the given time $\alpha^y$ is enough to travel $p_i$ completely, repairman $r_i$ travels $p_i$ completely and stays at the last node to finish time $\alpha^y$ and then comes back to $s_i$\footnote{In fact, if at round $y$ repairman $r_i$ comes back to $s_i$ as soon as it finishes traveling $p_i$ results in a better total latency in some cases. We avoid this because it is harder to analyze and explain. Moreover, in the worst case the total latency remains the same.}.

Now we prove that if a client is served indirectly with latency $q$ in $sol$ it will be served perfectly with latency $(3 + \epsilon)q$ in $sol'$. Suppose an arbitrary client $c_j$ is served indirectly with $r_i$ at time $q$ in a node $u_j$ of path  $p_i$. Note that $q$ either represents the time when both $r_i$ and $c_j$ arrive at $u_j$ or the time when $c_j$ arrives at $u_j$ but $r_i$ is already passed $u_j$. In $sol'$, when $c_j$ reaches $u_j$ it waits for repairman $r_i$ to visit $u_j$ after or at time $q$ during its back and forth travels.

Note that each round $x$ takes $2 \alpha^{x}$ units of time. Repairman $r_i$ can serve $c_j$ perfectly the first time it visits $u_j$ after time $q$. The first $\lfloor \log_{\alpha} q \rfloor$ rounds take $\sum_{x = 0}^{\lfloor \log_{\alpha} q \rfloor} 2 \alpha^{x}$ units of time which is equal to $2(\alpha^{\lfloor \log_{\alpha} q \rfloor + 1} - 1)/(\alpha - 1)$ and hence greater than $q$. Therefore $r_i$ serves $c_j$ perfectly at most at round $\lfloor \log_{\alpha} q \rfloor + 1$. At round $\lfloor \log_{\alpha} q \rfloor + 1$, repairman $r_i$ needs at most another $q$ units of time to reach to $u_j$. Therefore when $r_i$ travels at most $\sum_{x = 0}^{\lfloor \log_{\alpha} q \rfloor} 2 \alpha^{x} + q$ units of time, it visits $u_j$ and serves $c_j$ perfectly. Thus, the latency of client $c_j$ getting served perfectly is at most $\sum_{x = 0}^{\lfloor \log_{\alpha} q \rfloor} 2 \alpha^{x} + q \leq \left(3 + \frac{2}{\alpha - 1}\right) q$ which is equal to $(3 + \epsilon) \cdot q$ by replacing back $\alpha = 1 + \frac{2}{\epsilon}$. The lemma follows by applying the same argument to all the clients.
\qed
\end{proof}

We focus on finding a solution to the \smr{} problem where the clients are served indirectly and then transform it to a solution which serves the clients perfectly using \cref{serving-method}. Therefore, from now on whenever we use serving we mean serving indirectly.

We start with some important definitions.
\begin{definition}
Let neighborhood $\B(c, t_c)$ denote the set of all nodes whose distances from $s'_c$, the starting node of client $c$, are at most $t_c$.
\end{definition}

\begin{definition}
Let $\Po(r, t_r)$ denote the set of all non-simple paths (\ie they can visit nodes or edges multiple times) with length at most $t_r$ starting from $s_r$, the starting depot of $r$.
\end{definition}

Using the above two definition we can formalize the notion of serving as follows.
\begin{definition}
We call a repairman $r$ serves client $c$ or client $c$ getting served by $r$ at time $t$ if the path selected for $r$ hits neighborhood $\B(c, v'_c \cdot t)$ where $v'_c$ is the speed of client $c$.
\end{definition}

Let $mv$ be the maximum speed of all the clients and repairmen. We multiply all the edges of the graph by $2 \cdot mv$ which scales all the service times by factor $2 \cdot mv$. Now we can assume that the minimum service time a client can see is at least $1$.
Let $T$ be the largest service time a client can see, here we upper bound $T$ to be $\frac{2 \cdot MST(G)}{\min_i{v_i}}$ which is the units of time required to travel all the edges by the slowest repairman and hence serving all the clients.
We use set $Q = \{ 1, 2, \ldots, 2^i, \ldots, 2^{\left\lceil \log T \right\rceil + \left\lceil \log m \right\rceil / 2  + 1}\}$ to index geometrically increasing time-stamps. The greatest element of $Q$ is chosen such that all the clients are guaranteed to be served by our algorithm after this time-stamp. Note that we have $\left\lceil \log T \right\rceil + \left\lceil \log m \right\rceil / 2  + 1$ elements in $Q$ and hence its size is polynomially bounded by the size of input.

\subsection{LP formulation for \smr{}}
\label{sec:lp}
In this section we introduce an LP formulation for the \smr{} problem and show how to solve this LP approximately. We use the following LP for \smr{} inspired by the ideas from LPs introduced by Chakrabarty and Swamy \cite{CS10}.
\begin{align}
&\min &\sum_{c \in C} \sum_{t \in Q} &t \cdot y_{c,t}&\tag{{\bf PLP}} \label{plp}\\
&s.t.&\sum_{p \in \Po(r,v_r \cdot t)} x_{r,p,t} &\leq 1 &\forall t \in Q, \forall r \in R \label{plp1} \\
&&\sum_{r \in R}\ \sum_{p \in \Po(r, v_r \cdot t): p \cap \B(c, v'_c \cdot t) \neq \emptyset} x_{r,p,t} &\geq \sum_{t' \leq t} y_{c,t'} &\forall c \in C, \forall t \in Q  \label{plp2} \\
&&\sum_{t \in Q} y_{c, t} &\geq 1 &\forall c \in C \label{plp3}\\
&&x,y &\geq 0
\end{align}

The variable $x_{r,p,t}$ is the indicator variable showing whether repairman $r$ travels path $p\in \Po(r, v_r \cdot t)$ completely by time $t$. Note that if $p$ is in set $\Po(r, v_r \cdot t)$, from the definition of $\Po(r, v_r \cdot t)$, $r$ can complete traveling $p$ within time $t$. Variable $y_{c,t}$ is the indicator variable showing if client $c$ is served at time $t$.

Constraints (\ref{plp3}) guarantee that every client gets served. Constraints (\ref{plp1}) require each repairman $r$ to travel at most one path by time $t$. The amount $\sum_{t' \leq t} y_{c,t'}$ shows the fraction of service, client $c$ demands until time $t$ and the amount $\sum_{r \in R}\ \sum_{p \in \Po(r, v_r \cdot t): p \cap \B(c,v'_c \cdot t) \neq \emptyset} x_{r,p,t}$ shows the fraction of service $c$ gets from the repairmen until time $t$. Constraints (\ref{plp2}) guarantee that the total service from all the repairmen is at least as large as the demand from the client $c$. Note that Constraints (\ref{plp2}) just require the total demand by client $c$ should be served by the repairmen at anytime before time $t$ which is the case for serving indirectly.

Note that we only consider times that are in the set $Q$ but in a solution of \smr{} the clients may be served at any time.
\begin{lemma}
\label{lem:2approx}
The optimal value of \ref{plp} is at most twice the optimal solution of \smr{}.
\end{lemma}
\begin{proof}
Before proving the lemma, note that we assumed that the smallest element of $Q$ (the smallest latency seen by the clients) is one. Remember it does not change our analysis since the time-stamps in $Q$ grow exponentially which guarantees that the size of $Q$ is polynomial in terms of the inputs.

Intuitively the factor two comes from the fact that we only consider the powers of $2$ for the time-stamps in set $Q$, because if a client must served at time $t$ in an optimal solution in the LP it might wait for the next power of two to get served which can be at most $2t$. More formally, we show every integral solution ($\hat{sol}$) to the \smr{} problem with total latency $\ell^*$ can be transformed to a feasible solution $(\hat{x}, \hat{y})$ to \ref{plp} with the objective value at most $2\ell^*$. If repairman $r_i$ travels path $p_i$ by time $t_i$ in $\hat{sol}$ we set $\hat{x}_{r_i, p_i, t_i} = 1$. If client $c$ is served at time $t$ in $\hat{sol}$, we set $y_{c,t'} = 1$ where $t' = 2^{\left\lceil \log t \right\rceil}$. We set all the other non-set entries of $(\hat{x}, \hat{y})$ to zero. The objective value for $(\hat{x}, \hat{y})$ is at most twice the total latency of $\hat{sol}$ because if a client ($c$) is served at time $t$ in $\hat{sol}$, $c$ contributes at most $t'$ in the objective value for $(\hat{x}, \hat{y})$ where $t' < 2t$ since $2^{\left\lceil \log t \right\rceil} < 2^{\log t +1}  $. Thus the optimal value of \ref{plp} is at most twice the total latency of an optimal solution for \smr{}.
\qed
\end{proof}

In the next subsection we show that we are able to find a solution to the \smr{} problem which has total latency at most $O(\log n)$ times the optimum value of \ref{plp} which along with \cref{lem:2approx} upper bounds the integrality gap of the LP, where $n = |V|$ is the number of nodes in the metric space.
\subsection{Solving \ref{plp} in polynomial time}
\label{solving-plp}
The first difficulty to solve \ref{plp} is that it has exponentially many variables. In order to solve the LP we formulate its dual. The dual LP has exponentially many constraints but polynomially many variables therefore we need a separation oracle for the constraints in order to solve the dual LP in polynomial time.
The dual LP for \ref{plp} is as follows.
\begin{align}
&\max&\sum_{c \in C} \lambda_c &- \sum_{r \in R,t \in Q} \beta_{r,t}&& \tag{\textbf{DLP}} \label{dlp}\\
&s.t.&\sum_{c: p \cap \B(c, v'_c \cdot t) \neq \emptyset} \theta_{c,t} &\leq \beta_{r,t} & \forall r \in R, \forall t \in Q, \forall p \in \Po(r, v_r \cdot t)  \label{pcc} \\
&&\lambda_c \leq t +&  \sum_{t \leq t'} \theta_{c,t'}  &\forall c \in C, \forall t \in Q \label{dlp2}\\
&& \lambda, \beta, \theta &\geq 0 \label{dlp3}
\end{align}
We have exponentially many Constraints (\ref{pcc}), therefore we need a separation oracle for them in order to use Ellipsoid algorithm to solve \ref{dlp}. Given a candidate solution $(\lambda, \beta, \theta)$ for any repairman $r_i \in R$ and time-stamp $t \in Q$ we define Separation Oracle Problem $SOP(r_i, t)$ as follows.
Assume that each client $c$ has profit $\theta_{c,t}$ and  $\B$-ball $\B(c, v'_c \cdot t)$. The objective is to find a path in $\Po(r_i, v_i \cdot t)$ (has maximum length $t \cdot v_i$) which collects the maximum profit where a path collects the profit of any client whose $\B$-ball is hit by the path. If for all $r_i \in R$ and $t \in Q$ the optimal path collects at most $\beta_{r,t}$ profits, there is no violating constraint and $(\lambda, \beta, \theta)$ is a feasible solution; otherwise there exists a separating hyperplane.

The separation oracle explained above is NP-Hard since it contains the orienteering problem as a special case where the radius of all the $\B$-balls are zero. Therefore, we can only hope for an approximate solution for the separation oracle unless $P = NP$.

Note that $SOP(r_i, t)$ is the same as instance $(V, d, r_i, C, t \cdot v_i)$ of the NPCST problem (Definition \ref{npcst-def}) except instead of finding an optimum tree we have to find an optimum path. Because paths are the special cases of the trees, the optimum value for the NPCST instance is at least the optimum value of $SOP(r_i, t)$. Therefore, if we solve the NPCST instance we collect at least the same amount of profit. We will use the $\big( O(\log n), O(\log n), 2 \big)$-approximation algorithm in Theorem \ref{npcst-thm} to solve the NPCST instance and transform the resulting tree to a path by doubling the edges and taking an Eulerian tour which increases the length of the path by a factor of $2$. In fact, we approximately solve $SOP(r_i, t)$ by violating the budget on the resulting path, the radius of clients' $\B$-balls, and not collecting the maximum profit.

Due to all the violations explained above on the constraints of $SOP(r_i, t)$ we cannot bound the objective value of the feasible solution resulting from the $\big( O(\log n), O(\log n), 2 \big)$-approximation algorithm.
To this end, we introduce a relaxation of \ref{plp} (\ref{rplp}) in the following, when $\mu, \omega$ are constant integers greater than or equal to $1$.

\begin{align}
&\min &\sum_{c \in C} \sum_{t \in Q} &t \cdot y_{c,t}&\tag{\textbf{$PLP^{(\mu, \omega)}$}} \label{rplp}\\
&s.t.&\sum_{p \in \Po(r,\mu \cdot v_r \cdot t)} x_{r,p,t} &\leq \omega &\forall r \in R, \forall t \in Q \label{rplp1} \\
&&\sum_{r \in R} \sum_{p \in \Po(r, \mu \cdot v_r \cdot t): p \cap \B(c,\mu \cdot v'_c \cdot t) \neq \emptyset} x_{r,p,t} &\geq \sum_{t' \leq t} y_{c,t'} &\forall c \in C, \forall t \in Q  \label{rplp2} \\
&&\sum_{t \in Q} y_{c, t} &\geq 1 &\forall c \in C \label{rplp3}\\
&&x,y &\geq 0
\end{align}

\cref{rplp1} is the same as \cref{plp1} except instead of $\Po(r, v_r \cdot t)$ we have $\Po(r,\mu \cdot v_r \cdot t)$ which allows repairman $r$ to take a path which is $\mu$ times longer than a regular path in $\Po(r, v_r \cdot t)$. Moreover by putting $\omega$ instead of $1$ we allow each repairman to take $\omega$ routes instead of one. \cref{rplp2} is the same as \cref{plp2} except instead of $\B(c,v'_c \cdot t)$ we have $\B(c,\mu \cdot v'_c \cdot t)$ and instead of $\Po(r, v_r \cdot t)$ we have $\Po(r,\mu \cdot v_r \cdot t)$ which allow both the repairmen and clients to take paths that are $\mu$ times longer.

In the following lemma we show a $(\sigma, \phi, \omega)$-approximation algorithm for NPCST can be used to find a feasible solution to \ref{rplp} whose cost is at most the optimum solution of \ref{plp}. The proof of this lemma which is provided in full  is relatively involved and is
 more general than a lemma used in \cite{CS10}.
\begin{lemma}
\label{relaxation}
Given a $(\sigma, \phi, \omega)$-approximation algorithm for NPCST, one can find a feasible solution to \ref{rplp} in polynomial time, where $\mu = \max(\sigma, 2\cdot \phi)$, with objective value at most $\opt (1 + \epsilon)$ for any $\epsilon > 0$ where $\opt$ is the optimal value of \ref{plp}.
\end{lemma}

\begin{proof}
\label{missingproofs-relaxation}
Consider the following polytope (DLP($\xi; \mu, \omega$)).
\begin{align}
&&\sum_{c \in C} \lambda_c - \sum_{r\in R,t \in Q} \beta_{r,t} & \geq \xi&& \label{rdlp1}\\
&&\sum_{c: p \cap \B(c, \mu \cdot v'_c \cdot t) \neq \emptyset} \theta_{c,t} &\leq \frac{1}{\omega} \beta_{r,t} & \forall r \in R, \forall t \in Q, \forall p \in \Po(r,\mu \cdot v_r \cdot t)  \label{rdlp2} \\
&&\lambda_c \leq t +&  \sum_{t \leq t'} \theta_{c,t'}  &\forall c \in C, \forall t \in Q \label{rdlp3}\\
&& \lambda, \beta, \theta &\geq 0 \label{rdlp4}
\end{align}

DLP($\xi; \mu, \omega$) is feasible if the optimal value of the dual LP for \ref{rplp} is greater than or equal to $\xi$ since Constraints (\ref{rdlp2}), (\ref{rdlp3}), and (\ref{rdlp4}) are the constraints of the dual of \ref{rplp} and Constraint (\ref{rdlp1}) lower bounds the objective value of the dual LP. Therefore DLP($\xi; 1, 1$) is feasible if \ref{dlp} has the optimum value of greater than or equal $\xi$.

First we prove the following claim.
\begin{nclaim}
\label{cl:1}
Given a real value $\xi$ and triple $(\beta, \lambda, \theta)$ as the candidate solution to DLP($\xi; 1, 1$), there is a polynomial time separation oracle that either:
(1) shows $(\beta, \lambda, \theta) \in$ DLP($\xi; 1, 1$), or (2) finds a hyperplane separating $(\beta, \lambda, \theta)$ and DLP($\xi; \mu, \omega$).
\end{nclaim}
\begin{proof}
First we check if triple $(\beta, \lambda, \theta)$ satisfies all \cref{rdlp1,rdlp3,rdlp4}. The checks can be done in polynomial time as there are polynomially many \cref{rdlp1,rdlp3,rdlp4}. If a constraint does not satisfy, then we find a hyperplane separating $(\beta, \lambda, \theta)$ and DLP($\xi; \mu, \omega$) and the claim follows.

We might have exponentially many Constraints (\ref{rdlp2}). In order to check if all of Constraints (\ref{rdlp2}) are satisfied, for every value $t \in Q$ and each repairman $r \in R$ we define instance $\I_{t,r} = (V, d, s_r, C, v_r \cdot t) $ of NPCST (see \cref{npcst-def}) as follows. Node set $V$ and metric $d$ in $\I_{t,r}$ is the same as graph $G$ in the input of \smr{}, the root node $s_r$ is the starting depot of $r$, the cost budget for the tree is $v_r \cdot t$ ($v_r$ is the speed of $r$), and each client $c$ has profit $\theta_{c,t}$ and neighborhood $\B(c, v'_c \cdot t)$. The separation oracle is the following. We run the given $(\sigma, \phi, \omega)$-approximation algorithm on $\I_{t,r}$ and find a tree whose cost is at most $\phi \cdot v_r \cdot t$ and collects at least $\frac{1}{\omega}$ fraction of the optimum profit while violating the radius of the clients' $\B$-ball by $\sigma$ factor. We transform the resulting tree to a path by doubling the edges and taking an Eulerian tour which makes the length of the tour to be at most $2 \cdot \phi \cdot v_r \cdot t$. Note that because $\mu = \max(\sigma, 2\cdot \phi)$ the resulting tour is in $\Po(r,\mu \cdot v_r \cdot t)$ and it collects the profit of client $c$ by visiting a node in its $\B(c, \mu \cdot v'_c \cdot t)$ neighborhood ball.

If there exits repairman $r \in R$ and time $t\in Q$ such that the path resulting from the separation oracle collects profits greater than $\frac{1}{\omega} \beta_{r,t}$, the corresponding Constraint (\ref{rdlp2}) gives a separating hyperplane and the claim follows. If not, we prove by contradiction that $(\beta, \lambda, \theta) \in$ DLP($\xi; 1, 1$). If $(\beta, \lambda, \theta) \not\in$ DLP($\xi; 1, 1$) then at least one of the constraints of DLP($\xi; 1, 1$) has to not hold for $(\beta, \lambda, \theta)$. As we check all \cref{rdlp3,rdlp4} for DLP($\xi; \mu, \omega$) and they are the same in DLP($\xi; 1, 1$) the violating constraint ($\V$) is one of the \cref{rdlp2}. Let the violating constraint $\V$ happens for repairman $r$ and time $t$, and $p^*_{r,t}$ be the path we found by running our separation oracle algorithm on $\I_{t,r}$. Constraint $\V$ being a violating constraint means that there exists a path with length at most $v_r \cdot t$ which collects profits greater than $\beta_{r,t}$ without any violation in the clients' $\B$-ball. Therefore, as $p^*_{r,t}$ is the path resulting from the $(\sigma, \phi, \omega)$-approximation algorithm, $p^*_{r,t}$ is in $\Po(r,\mu \cdot v_r \cdot t)$ (remember $\mu$ is at least $2 \cdot \phi$), collects profits at least  $\frac{1}{\omega} \cdot \beta_{r,t}$ while violating the clients' $\B$-ball by a factor of at most $\mu$ (remember $\mu$ is at least $\sigma$). This means $r$, $t$, and $p^*_{r,t}$ is a hyperplane separating $(\beta, \lambda, \theta)$ and DLP($\xi; \mu, \omega$) which cannot happen. Therefore there is no such violating constraint $\V$ and hence $(\beta, \lambda, \theta)$ is in polytope DLP($\xi; 1, 1$).
\qed
\end{proof}

The following proposition is obtained by running Ellipsoid algorithm \cite{GLS93} for the separation oracle of \cref{cl:1}.
\begin{proposition}
\label{elipsoid}
For any value $\xi$, Ellipsoid algorithm in polynomial time either shows DLP($\xi; \mu, \omega$) is empty or finds a feasible solution $(\beta, \lambda, \theta)$ in the polytope DLP($\xi; 1, 1$).
\end{proposition}
We find the largest value ($\xi^*$), within a factor of $(1 + \epsilon)$, for which there exists a feasible solution $(\beta^*, \lambda^*, \theta^*) \in$ DLP($\xi^*; 1, 1$) by binary search over $\xi$ and \cref{elipsoid}. Because DLP($\xi^*; 1, 1$) is feasible, it means that \ref{dlp} has the optimum value at least $\xi^*$ (since DLP($\xi^*; 1, 1$) has the same constraints as (\ref{dlp}) with an extra constraint to lower bound the objective value). Thus, from the LP duality theorem we conclude $\xi^* \leq \opt$ (Fact 1).

Because $\xi^*$ is the largest value within a factor of $(1 + \epsilon)$ for which DLP($\xi; 1, 1$) is non-empty, Ellipsoid algorithm in \cref{elipsoid} with $\xi = \xi^*( 1 + \epsilon)$ terminates in polynomial time certifying infeasibility of DLP($\xi^*(1  + \epsilon); \mu, \omega$). Thus it generates a collection of constraints of type (\ref{rdlp2}), (\ref{rdlp3}), (\ref{rdlp4}), and (\ref{rdlp1}) which all together constitute an infeasible system of constraints. Lets denote this infeasible system of constraints by $\tau$. Note that $\tau$ consist of polynomially many constraints as Ellipsoid algorithm of \cref{elipsoid} runs in polynomial time and at each step it finds one violating constraint.  We apply the Farkas's lemma to the constraints in $\tau$.
\begin{lemma}[Farkas]
\label{farkas}
For a matrix $A\in \bbR^{n \times m}$ and a vector $b \in \bbR^m$, exactly one of the following holds.
\begin{enumerate}
\item $\exists w \in \bbR^m \geq 0$ such that $Aw \geq b$.
\label{farkas1}
\item $\exists z \in \bbR^n \geq 0$ such that $A^Tz \leq 0$ and $b^Tz > 0$.
\label{farkas2}
\end{enumerate}
\end{lemma}
Here we represent the constraints of type (\ref{rdlp2}), (\ref{rdlp3}), and (\ref{rdlp4}) in $\tau$  by $A^Tz \leq 0$ and Constraint (\ref{rdlp1}) in $\tau$ by $b^Tz > 0$. Therefore as the constraints in $\tau$ are infeasible, Case \ref{farkas2} of Farkas' lemma is false which implies Case \ref{farkas1} is true. Because $Aw \geq b$ is the dual constraints of $A^Tz \leq 0$ which is $\tau$, $Aw \geq b$ is just constraints of $\ref{rplp}$. Thus, Case \ref{farkas1} of Farkas' lemma implies existences of a solution $(x, y)$ that is feasible for \ref{rplp} with objective value at most $\xi^* (1 +\epsilon)$. Because the number of constraints in $\tau$ (size of $A$ in Case \ref{farkas1} of Farkas Lemma) is polynomially bounded, we can actually find a feasible solution $(x, y)$ to \ref{rplp} with objective value at most $\xi^*(1 + \epsilon)$ by Ellipsoid algorithm.
From Fact 1 we know that $\xi^* \leq \opt$, therefore the objective value of solution $(x, y)$ to \ref{rplp} is at most $\opt (1 + \epsilon)$.
\qed

\end{proof}

\subsection{Rounding the LP}
\label{sec:rounding}
We show how to use feasible solution $(x, y)$ taken from \cref{relaxation} to obtain an integral solution to \smr{} with the total latency at most $O(\max(\sigma, 2\phi) \cdot \omega \cdot \opt)$ and thus finish the proof of Theorem \ref{general-thm}. Sum-Movement Repairmen Algorithm (\mra{}) shown in \cref{mra} is our algorithm to do so.

\begin{figure}[h]
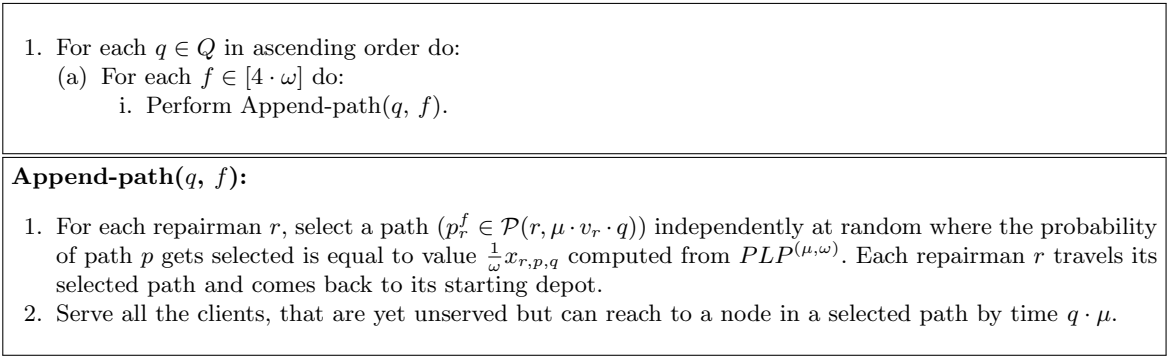

\fbox{\parbox{\textwidth}{
\begin{enumerate}
\item For each $q \in Q$ in ascending order do:
\begin{enumerate}
\item For each $f \in [4 \cdot \omega]$ do: \label{algq}
\begin{enumerate}
\item Perform \stp{q}{f}.
\end{enumerate}
\end{enumerate}
\end{enumerate}
}}
\fbox{\parbox{\textwidth}{
{\bf \stp{q}{f}:}
\begin{enumerate}
\item For each repairman $r$, select a path ($p^f_{r}\in \Po(r, \mu\cdot v_r \cdot q)$) independently at random where the probability of path $p$ gets selected is equal to value $\frac{1}{\omega}x_{r,p,q}$ computed from \ref{rplp}. Each repairman $r$ travels its selected path and comes back to its starting depot. \label{random-selection}
\item Serve all the clients, that are yet unserved but can reach to a node in a selected path by time $q \cdot \mu$. \label{client-cover}
\end{enumerate}
}}
\caption{Movement Repairmen Algorithm (\mra{})}
\label{mra}
\end{figure}

As explained earlier $Q$ in \mra{} is the set $\{ 1, 2, \ldots, 2^i, \ldots, 2^{\left\lceil \log T \right\rceil + \left\lceil \log m \right\rceil / 2  + 1}\}$  where $T$ is the latest service time a client can see which was upper bounded by $\frac{2 \cdot MST(G)}{\min_i{v_i}}$ and the value $\left\lceil \log T \right\rceil + \left\lceil \log m \right\rceil / 2  + 1$ is chosen to guarantee that \mra{} serves all the clients after it finishes. Moreover, $\mu$ and $\omega$ are the constants in \ref{rplp}.

\mra{} serves clients in multiple steps. It starts serving clients with paths that have the maximum latency $1 \cdot \mu$ then it concatenate paths of maximum latency $2 \cdot \mu$, then $4 \cdot \mu$ and so on. These paths come from the set $\Po(r,\mu \cdot v_r \cdot q)$ for $q \in Q$ and the selection is done using \ref{rplp} variables $x_{r,p,q}$. In fact, for each $q \in Q$ we select $\fw$ paths by executing \stp{q}{f} $\fw$ times where $f$ is just used to iterate over set $[\fw]$. This is because we want to have independence between the paths selected at each execution of \stp{q}{f} which helps us to better analyze the number of clients get served in the execution.

We use the following definitions to refer to the clients served by \mra{}.
\begin{definition}
\label{seta}
Let $A^{q,f}$ denotes the set of non-served clients getting served at Instruction \ref{client-cover} of \stp{q}{f}. \ie $A^{q,f}$ is the set of clients $c$ such that $c$ is not served before the execution of \stp{q}{f} but it can reach a node $v$ by time $\mu \cdot q$ such that there exists a repairman $r$ with $v \in p^f_r$ (remember $p^f_r$ is the path selected for $r$ in \stp{q}{f}).
\end{definition}
\begin{definition}
\label{setaa}
Set $\A^{q,f} = \bigcup_{(q',f') \leq (q,f)} A^{q',f'}$ is the set of all clients served by \mra{} up to and including the execution of \stp{q}{f}. Here the operator $\leq$ is the lexicographic ordering for the ordered pairs where the first entry has more priority than the second one.
\end{definition}
We define function $prev(q,f)$ as follows.
\[
prev(q,f) = \left\{
  \begin{array}{l l}
    (q,f-1) & \quad f\neq 1\\
    (\frac{q}{2},4 \cdot \omega) & \quad f = 1
  \end{array} \right.
\]
\begin{definition}
\label{valuef}
Let $(q', f') = prev(q,f)$ and \stp{q'}{f'} be the predecessor of \stp{q}{f}. Let $F^{q,f}$ denote the value of $\sum_{c \in C \setminus \A^{q',f'}} \sum_{t\leq q} y_{c,t}$. Intuitively, $F^{q,f}$ can be taught as the fractional number of clients that are (fractionally) served in feasible solution $(x,y)$ by the time $q$, but not served by \mra{} before the execution of \stp{q}{f}.
\end{definition}
We would like in $A^{q,f}$, be a large fraction of $F^{q,f}$.
First we prove the following lemma to lower bound the probability of a client getting served in the execution of $\stp{q}{f}$.
\begin{lemma}
\label{prob-cover-lem}
Let $q$ be any element of $Q$ and $c$ be any client in $C$. If we randomly select a path for each repairman $r \in R$ such that the probability of selecting $p \in \Po(r, \mu \cdot v_r \cdot q)$ is $\frac{1}{\omega}\cdot x_{r,p,q}$, then the probability of $c$ getting served (a selected path visits a node from $\B(c, \mu q)$) is at least $\frac{1}{2 \omega} \cdot \sum_{q' \leq q} y_{c,q'}$.
\end{lemma}

\begin{proof}
The probability of a client $c$ getting served by an arbitrary repairman $r \in R$ is $D_r = \frac{1}{\omega} \sum_{p \in \Po(r, \mu \cdot v_r \cdot t): p \cap \B(c,\mu \cdot v'_c \cdot t) \neq \emptyset} x_{r,p,t}$ from the probability distribution used in the rounding. To simplify the notations, let $B = \sum_{r \in R'} D_r$ and $Y_c = \frac{1}{\omega} \sum_{q' \leq q} y_{c,q'}$.

The probability that a client $c$ is not served by any repairman in $R$ is $\prod_{r \in R} (1 - D_r)$.
\begin{align*}
\prod_{r \in R} (1 - D_r) & \leq \left(\frac{{|R|} - \sum_{r \in R} D_r}{|R|}\right)^{|R|} & \quad \text{Arithmetic and Geometric}\\
&&\text{Means Inequality \footnotemark}\\
& = \left(1 - \frac{B}{|R|}\right)^{|R|} & (\text{replacing by }B )\\
& = \left(1 - \frac{1}{\frac{|R|}{B}}\right)^{\frac{|R|}{B}B } & \\
& \leq e^{-B} & \\
& \leq e^{-Y_c} & \text{\cref{rplp2} }\\
\end{align*}
\footnotetext{For any set of $n$ non-negative numbers $x_1, \ldots, x_n$ we have $\frac{x_1 + \ldots + x_2}{n} \geq \sqrt[n]{x_1 \cdot x_2 \cdot \ldots \cdot x_n}$}
From the above inequality we conclude that client $c$ gets served with probability at least  $1 - e^{-Y_c}$. The following inequalities finish the proof of the lemma.
\begin{align*}
1 - e^{-Y_c} & = 1 - \left(\sum_{i = 0}^{\infty} \frac{(-Y_c)^i}{i!}\right)  & \text{by Taylor Expansion}\\
& \geq Y_c - \frac{Y_c^2}{2}& \text{as $0 \leq Y_c \leq 1$}\\
& \geq \frac{1}{2} Y_c & \text{as $0 \leq Y_c \leq 1$}\\
& \geq \frac{1}{2 \omega} \cdot \sum_{q' \leq q} y_{c,q'} & \text{definition of $Y_c$}
\end{align*}
\qed
\end{proof}

We use the following lemma to derandomize selections of the paths in \stp{q}{f} and to show that $A^{q,f}$ is at least $\left\lceil \frac{F^{q,f}}{2 \cdot \omega} \right\rceil$.
\begin{lemma}
\label{coverage}
We can derandomize \stp{q}{f} to deterministically select a path $p^f_r \in \Po(r, q \cdot v_r \cdot \mu)$ for each repairman $r$, such that the set of newly served clients ($A^{q,f}$ as defined in \cref{seta}) to be at least $\left\lceil \frac{F^{q,f}}{2 \cdot \omega} \right\rceil$.
\end{lemma}
\begin{proof}
If we select a path $p^f_r \in \Po(r, \mu \cdot v_r \cdot q)$ with probability $x_{r,p^f_r,q} / \omega$ for each repairman $r$, from Lemma \ref{prob-cover-lem} we know that the probability of an arbitrary client $c$ getting served is at least $\frac{1}{2 \omega} \cdot \sum_{t \leq q} y_{c,t}$. By linearity of the expectation we conclude the expected number of clients served with these paths is at least $\frac{1}{2 \omega} \cdot \sum_{c \in C} \sum_{t \leq q} y_{c,t}$. Therefore the expected number of new clients that are served with these paths is at least $\frac{F^{q,f}}{2 \omega}$ by definition of $F^{q,f}$ (see \cref{valuef}).

Now, we derandomize the random selection of the paths in \stp{q}{f} so that we serve deterministically at least $\left\lceil \frac{F^{q,f}}{2 \omega } \right\rceil$ number of new clients.
Let $R'\subseteq R$ be an arbitrary subset the set of the repairmen. For each client $c \in C$ and time $q \in Q$ we define variable $Y^{R'}_{c,q}$ as follows.
\begin{equation}
\label{cy}
Y^{R'}_{c,q} = \min\left(\sum_{q' \leq q} y_{c,q'}, \frac{1}{\omega} \sum_{r \in R'}\ \sum_{p \in \Po(r, \mu \cdot v_r \cdot t): p \cap \B(c,\mu \cdot v'_c \cdot t) \neq \emptyset} x_{r,p,t}\right)
\end{equation}
Intuitively, $Y^{R'}_{c,q}$ represents the fractional service client $c$ receives from the Repairmen in $R'$ in solution $(x,y)$, if each repairman could take at most one path. Here, we use $\frac{1}{\omega}$ because in $(x,y)$ each repairman can take up to $\omega$ paths (see \cref{plp1} in \ref{plp}).
Consequently, for an arbitrary subset $C' \subseteq C$, $R' \subseteq R$ and, time $q \in Q$, we define the variable $\Y^{C',R'}$ as follows.
\begin{equation}
\label{ccy}
\Y^{C', R'}_q = \sum_{c \in C'} Y^{R'}_{c,q}
\end{equation}
Intuitively, we can think of $\Y^{C', R'}_q$ to be the (fractional) amount of service that clients in $C'$ receive from the repairmen in $R'$ by time $q$ in feasible solution $(x,y)$. We prove the following claim.
\begin{nclaim}
\label{claim:derandomize}
Let $C'$ be an arbitrary subset of $C$ and $R'$ be an arbitrary subset of $R$. We can deterministically select one path for each repairman in $R'$ such that they serve at least $\left\lceil \frac{1}{2} \Y^{C', R'}_q \right\rceil$ clients from $C'$.
\end{nclaim}
\begin{proof}
We prove this claim by induction on $|R'|$. For the base case when $|R'| = 1$ assume that repairman $r$ is the only member of $R'$. For each path $p \in \Po(r, \mu \cdot v_r \cdot q)$ let $C'_p$ be the set of clients $c$ in $C'$ whose neighborhood ball ($\B(c, v'_c \cdot \mu \cdot t)$) gets hit by $p$.  Let path $p_r^* \in \Po(r, \mu \cdot v_r \cdot q)$ be the path that intersects with the maximum number of neighborhood balls of the clients in $C'$ (\ie $|C'_{p^*}|$ is maximum). We select path $p^*$ for $r$. Consider the following inequalities.
\begin{align*}
\Y^{C', \{r\}}_q &= \sum_{c \in C'} Y^{\{r\}}_{c,q} &\\
& = \sum_{c \in C'} \min\left(\sum_{q' \leq q} y_{c,q'}, \frac{1}{\omega} \sum_{p \in \Po(r, \mu \cdot v_r \cdot t): p \cap \B(c,\mu \cdot v'_c \cdot t) \neq \emptyset} x_{r,p,t}\right) & \text{definition of $Y^{C', \{r\}}_{c,q}$}\\
& \leq \frac{1}{\omega} \sum_{c \in C'}\  \sum_{p \in \Po(r, \mu \cdot v_r \cdot t): p \cap \B(c,\mu \cdot v'_c \cdot t) \neq \emptyset} x_{r,p,t} &\\
& \leq \frac{1}{\omega} \sum_{p \in \Po(r, \mu \cdot v_r \cdot t)}\ \sum_{c \in C' : p \cap \B(c,\mu \cdot v'_c \cdot t) \neq \emptyset} x_{r,p,t} &\\
& \leq \frac{1}{\omega} \sum_{p \in \Po(r, \mu \cdot v_r \cdot t)} |C'_p|   x_{r,p,t} & \text{definition of $C'_p$}\\
& \leq \frac{1}{\omega} \sum_{p \in \Po(r, \mu \cdot v_r \cdot t)} |C'_{p^*}|   x_{r,p,t} &\\
& \leq \frac{1}{\omega} \omega \cdot C'_{p^*}   &\text{\cref{rplp1}}\\
& \leq C'_{p^*}   &
\end{align*}
The above inequality shows $\Y^{C', \{r\}}_q \leq C'_{p^*}$ and hence $\left\lceil \frac{1}{2} \Y^{C', \{r\}}_q \right\rceil \leq C'_{p^*}$ as $C'_{p^*}$ is an integer value which completes the proof for the base case.

Assume that the claim holds for any subset of repairmen with size $k'$ as the induction hypothesis. We prove that the claim holds for an arbitrary subset $R' \subseteq R$ of size $k' + 1$. Let $r \in R'$ be a repairman in $R'$. Similar to the base case let path $p^* \in \Po(r, \mu \cdot v_r \cdot q)$ be the path that intersects with the maximum number of $\B$-balls of the clients in $C'$ (\ie $|C'_{p^*}|$ is maximum). We select path $p^*$ for $r$ which serves $|C'_{p^*}|$ new clients from $C'$. By the induction hypothesis we can select one path for each of remaining repairmen $R' \setminus\{r\}$ such that they serve at least $\left\lceil \frac{1}{2} \Y^{C'\setminus C'_{p^*}, R' - \{r\}}_q \right\rceil$ many clients. In the following we prove that $\left\lceil \frac{1}{2} \Y^{C', R'}_q \right\rceil \leq C'_{p^*} + \left\lceil \frac{1}{2} \Y^{C'\setminus C'_{p^*}, R' - \{r\}}_q \right\rceil$ which completes the proof of the induction for $R'$ and hence the claim.
\begin{align*}
&\left\lceil \frac{1}{2} \Y^{C', R'}_q \right\rceil\\
&= \left\lceil \frac{1}{2} \sum_{c \in C'} Y^{R'}_{c,q} \right\rceil &\\
&= \left\lceil \frac{1}{2} \left(\sum_{c \in C'\setminus C'_{p^*}} Y^{R'}_{c,q} + \sum_{c \in C'_{p^*}} Y^{R'}_{c,q}\right) \right\rceil &\\
&= \Bigg\lceil \frac{1}{2} \sum_{c \in C'\setminus C'_{p^*}} \min \left(\sum_{q' \leq q} y_{c,q'}, \frac{1}{\omega} \sum_{r' \in R'}\ \sum_{p \in \Po(r', \mu \cdot v_{r'} \cdot t): p \cap \B(c,\mu \cdot v'_c \cdot t) \neq \emptyset} x_{r',p,t}\right)\\
&+ \frac{1}{2} \sum_{c \in C'_{p^*}} Y^{R'}_{c,q} \Bigg\rceil & \text{def. of $Y^{R'}_{c,q}$}\\
&= \Bigg\lceil \frac{1}{2} \sum_{c \in C'\setminus C'_{p^*}} \min \Big(\sum_{q' \leq q} y_{c,q'}, \frac{1}{\omega} \sum_{r' \in R'\setminus\{r\}}\ \sum_{p \in \Po(r', \mu \cdot v_{r'} \cdot t): p \cap \B(c,\mu \cdot v'_c \cdot t) \neq \emptyset} x_{r',p,t}\\
&+ \frac{1}{\omega} \sum_{p \in \Po(r, \mu \cdot v_{r'} \cdot t): p \cap \B(c,\mu \cdot v'_c \cdot t) \neq \emptyset} x_{r,p,t}\Big)
+ \frac{1}{2} \sum_{c \in C'_{p^*}} Y^{R'}_{c,q} \Bigg\rceil & \text{excluding $r$ from $R'$}\\
&= \Bigg\lceil \frac{1}{2} \sum_{c \in C'\setminus C'_{p^*}} \min \left(\sum_{q' \leq q} y_{c,q'}, \frac{1}{\omega} \sum_{r' \in R'\setminus\{r\}}\ \sum_{p \in \Po(r', \mu \cdot v_{r'} \cdot t): p \cap \B(c,\mu \cdot v'_c \cdot t) \neq \emptyset} x_{r',p,t}\right)\\
&+ \frac{1}{2\omega} \sum_{c \in C'\setminus C'_{p^*}}\ \sum_{p \in \Po(r, \mu \cdot v_r \cdot t): p \cap \B(c,\mu \cdot v'_c \cdot t) \neq \emptyset} x_{r,p,t} + \frac{1}{2} \sum_{c \in C'_{p^*}} Y^{R'}_{c,q} \Bigg\rceil
\end{align*}
Here the last inequality results form the following property of the $\min$ function that for any $Z \geq 0$ and $X,Y \in \bbR$ we have $\min(X, Y + Z) \leq \min(X,Y) + Z$. In the following we continue with replacing the $\min$ function with $\Y^{C'\setminus C'_{p^*}, R' - \{r\}}_q$ by noting its definition in the Inequality (\ref{ccy}).
\begin{align*}
& \Bigg\lceil \frac{1}{2} \sum_{c \in C'\setminus C'_{p^*}} \min \left(\sum_{q' \leq q} y_{c,q'}, \frac{1}{\omega} \sum_{r' \in R'\setminus\{r\}}\ \sum_{p \in \Po(r', \mu \cdot v_{r'} \cdot t): p \cap \B(c,\mu \cdot v'_c \cdot t) \neq \emptyset} x_{r',p,t}\right)\\
&+ \frac{1}{2\omega} \sum_{c \in C'\setminus C'_{p^*}}\ \sum_{p \in \Po(r, \mu \cdot v_r \cdot t): p \cap \B(c,\mu \cdot v'_c \cdot t) \neq \emptyset} x_{r,p,t} + \frac{1}{2} \sum_{c \in C'_{p^*}} Y^{R'}_{c,q} \Bigg\rceil\\
&= \Bigg\lceil \frac{1}{2} \Y^{C'\setminus C'_{p^*}, R' - \{r\}}_q\\
&+ \frac{1}{2\omega} \sum_{c \in C'\setminus C'_{p^*}}\ \sum_{p \in \Po(r, \mu \cdot v_r \cdot t): p \cap \B(c,\mu \cdot v'_c \cdot t) \neq \emptyset} x_{r,p,t} + \frac{1}{2} \sum_{c \in C'_{p^*}} Y^{R'}_{c,q} \Bigg\rceil & \text{def. of $\Y^{C'\setminus C'_{p^*}, R' - \{r\}}_q$}\\
&\leq \Bigg\lceil \frac{1}{2} \Y^{C'\setminus C'_{p^*}, R' - \{r\}}_q\\
&+ \frac{1}{2\omega} \sum_{c \in C'\setminus C'_{p^*}}\ \sum_{p \in \Po(r, \mu \cdot v_r \cdot t): p \cap \B(c,\mu \cdot v'_c \cdot t) \neq \emptyset} x_{r,p,t} + \frac{1}{2} |C'_{p^*}| \Bigg\rceil & \text{Noting that $Y^{R'}_{c,q} \leq 1$}\\
&= \Bigg\lceil \frac{1}{2} \Y^{C'\setminus C'_{p^*}, R' - \{r\}}_q\\
&+ \frac{1}{2\omega} \sum_{p \in \Po(r, \mu \cdot v_r \cdot t)}\ \sum_{c \in C'\setminus C'_{p^*}: p \cap \B(c,\mu \cdot v'_c \cdot t) \neq \emptyset} x_{r,p,t} + \frac{1}{2} |C'_{p^*}| \Bigg\rceil & \text{reordering the}\\
&& \text{$sum$ functions}\\
&\leq \left\lceil \frac{1}{2} \Y^{C'\setminus C'_{p^*}, R' - \{r\}}_q
+ \frac{1}{2\omega} \sum_{p \in \Po(r, \mu \cdot v_r \cdot t)}\ |C'_{p^*}| x_{r,p,t} + \frac{1}{2} |C'_{p^*}| \right\rceil& \text{$p^*$ serves the maximum}\\
&& \text{ number of clients}\\
&\leq \left\lceil \frac{1}{2} \Y^{C'\setminus C'_{p^*}, R' - \{r\}}_q
+ \frac{1}{2\omega} \omega |C'_{p^*}| + \frac{1}{2} |C'_{p^*}| \right\rceil& \text{\cref{rplp1}}\\
&\leq \left\lceil \frac{1}{2} \Y^{C'\setminus C'_{p^*}, R' - \{r\}}_q\right\rceil + |C'_{p^*}|&\\
\end{align*}
\qed
\end{proof}
Similar to \cref{valuef} for value $F^{q,f}$, let $\A^{q', f'}$ be the set of served clients before execution of \stp{q}{f}.
We prove the following inequality about $\Y^{C', R'}_q$ where $R' = R$ and $C' = C \setminus \A^{q', f'}$ .
\begin{align*}
&\Y^{C \setminus \A^{q', f'}, R}_q \\
&= \sum_{c \in C \setminus \A^{q', f'}} \min\left(\sum_{q' \leq q} y_{c,q'}, \frac{1}{\omega} \sum_{r \in R}\ \sum_{p \in \Po(r, \mu \cdot v_r \cdot t): p \cap \B(c,\mu \cdot v'_c \cdot t) \neq \emptyset} x_{r,p,t}\right) & \text{Inequality (\ref{ccy})}\\
& \geq \sum_{c \in C \setminus \A^{q', f'}} \frac{1}{\omega} \sum_{q' \leq q} y_{c,q'} & \text{\cref{rplp2} of \ref{rplp}}\\
&=  \frac{F^{q,f}}{\omega} & \text{\cref{valuef}}
\end{align*}
The above inequality proves that $\left\lceil \frac{1}{2} \Y^{C \setminus \A^{q', f'}, R}_q \right\rceil \geq \left\lceil \frac{F^{q,f}}{2 \cdot \omega} \right\rceil$. Therefore by using Claim \ref{claim:derandomize} where $R' = R$ and $C' = C \setminus \A^{q', f'}$, we serve at least $\frac{F^{q,f}}{2 \cdot \omega}$ clients which finishes proof of the lemma.
\qed
\end{proof}

We prove the following lemma which combined with \cref{relaxation} finishes the proof of \cref{general-thm}.
\begin{lemma}
\label{rounding}
A feasible solution $(x, y)$ to the \ref{rplp} with objective value $\opt$ (the optimum value for \ref{plp}) can be rounded to an integral solution to \smr{} with total latency $O(\mu \cdot \omega) \cdot \opt$.
\end{lemma}
\begin{proof}
For the proof, we show \mra{} serves all the clients with total latency $O(\mu \cdot \omega) \cdot \opt$.
The following defintion is our last definition in this section.

\begin{definition}
\label{def-h}
For any $q \in Q$ we denote $h_q$ to be $\sum_{c \in C} y_{c,q}$. As $y_{c,q}$ denotes how much client $c$ is served in time $q$, $h_q$ can be thought as the total amount of fractional service clients receive at time $q$ in feasible solution $(x, y)$.
\end{definition}

Note that after executing \stp{q}{4 \cdot \omega} we jump to the next element in $Q$, \ie we execute \stp{2\cdot q}{1} and append paths of length $2q \cdot \mu$. Although our algorithm does not execute \stp{q}{\fw + 1}, we use value $F^{q, \fw + 1}$ (see \cref{valuef}) to denote the fractional number of clients that are served in feasible solution $(x,y)$ by time $q$ but are not served by \mra{} with paths of length at most $q \cdot \mu$ (executions of \stp{q}{f} for all $f$ in $[\fw]$). We us the following claim to upper bound $F^{q, \fw + 1}$.
\begin{nclaim}
\label{claim:stp}
For any $q = 2^a \in Q$ we have $F^{q, \fw + 1} \leq \sum_{s = 0}^a \left(\frac{1}{4^{a-s+1}}\right) h_{(2^s)}$.
\end{nclaim}
\begin{proof}
First we prove that for any $q \in Q$ we have $F^{q, \fw + 1} \leq \frac{F^{q,1}}{4}$.
By Lemma \ref{coverage} performing \stp{q}{f} serves at least $\left\lceil \frac{F^{q,f}}{2 \cdot \omega} \right\rceil$ new clients for any $f \in [4 \cdot \omega]$. In other words, $F^{q,f+1}$ is at most $F^{q,f}\left(1 - \frac{1}{2 \cdot \omega}\right)$ which implies that after the execution of $\stp{q}{f}$ we drop the number of non-served clients by at least a factor of $\left(1 - \frac{1}{2 \cdot \omega}\right)$. By iterating $f$ over set $[\fw]$ we conclude that after executing \stp{q}{4 \cdot \omega} the total number of non-served clients is at most $F^{q,1} \cdot \left(1 - \frac{1}{2 \cdot \omega}\right)^{4\cdot \omega}$ which concludes that $F^{q, \fw + 1} \leq \left(F^{q,1} \cdot \left(1 - \frac{1}{2 \cdot \omega}\right)^{4\cdot \omega}\right)$. Simplifying, we get

\begin{align}
\notag F^{q, \fw + 1} \notag & \leq F^{q,1} \cdot \left(1 - \frac{1}{2 \cdot \omega}\right)^{4\omega}& \text{since\ } \left(1 - \frac{1}{x}\right)^{x} \leq \frac{1}{e}\\
\notag & \leq F^{q,1} \cdot \left(\frac{1}{e}\right)^2& \\
\label{one-step} & \leq \frac{1}{4} \cdot F^{q,1}
\end{align}
We prove the claim by induction on $a$. The base case when $q' = 2^0$ is the direct result of Inequality (\ref{one-step}) by noting that $F^{1,1} = h_1$ (from the definition of $F^{1,1}$). Assume that for $q' = 2^{a'}$ the claim holds as the induction hypothesis. The following inequalities prove the claim's statement for $2q' = 2^{a' + 1}$ which finish the proof of the induction and hence the claim.
\begin{align*}
F^{2q', \fw + 1} &\leq \frac{1}{4} F^{2q',1}& \text{Due to Inequality (\ref{one-step}) with $q = 2q'$}\\
 &\leq \frac{1}{4} \left(F^{q', \fw + 1} + h_{2q'}\right) & \text{\cref{valuef} and \cref{def-h}}\\
 &\leq \sum_{0 \leq s \leq a'+1} \left(\frac{1}{4^{(a' + 1)-s+1}}\right) h_{(2^s)} & \text{Induction hypothesis}
\end{align*}
\qed
\end{proof}

Now we prove that \mra{} serves all the clients. Remember that $T$ was the latest time a client can be served which means that $h_t$ is equal to zero for all $t \geq T$. The greatest element of $Q$ is $2^{\left\lceil \log T \right\rceil + \left\lceil \log m \right\rceil / 2  + 1}$. After executing \stp{2^{\left\lceil \log T \right\rceil + \left\lceil \log m \right\rceil / 2 + 1 }}{4 \cdot \omega} we have at most the following number of clients to serve by Claim \ref{claim:stp}.

$$\sum_{0 \leq s \leq \left\lceil \log T \right\rceil + \left\lceil \log m \right\rceil / 2  +1 } \left(\frac{1}{4^{\left\lceil \log T \right\rceil + \left\lceil \log m \right\rceil / 2 -s+2}}\right) h_{(2^s)}$$
We show that this value is strictly less than one.
\begin{align*}
&\sum_{0 \leq s \leq \left\lceil \log T \right\rceil + \left\lceil \log m \right\rceil / 2 } \left(\frac{1}{4^{\left\lceil \log T \right\rceil + \left\lceil \log m \right\rceil / 2 -s+1}}\right) h_{(2^s)} &\\
&= \sum_{0 \leq s \leq \left\lceil \log T \right\rceil} \left(\frac{1}{4^{\left\lceil \log T \right\rceil + \left\lceil \log m \right\rceil / 2 -s+1}}\right) h_{(2^s)} & \text{Noting that $h_t = 0$ for $t \geq T$}\\
&\leq \sum_{0 \leq s \leq \left\lceil \log T \right\rceil} \left(\frac{1}{4^{\left\lceil \log m \right\rceil / 2 + 1}}\right) h_{(2^s)} & \\
&\leq \left(\frac{1}{4^{\left\lceil \log m \right\rceil/2 + 1}}\right) m  & \text{Noting that $\sum_{0 \leq s \leq \left\lceil \log T \right\rceil} h_{(2^s)} = m$}\\
&< 1
\end{align*}
Because \mra{} serves the clients integrally, after executing \stp{2^{\left\lceil \log T \right\rceil + \left\lceil \log m \right\rceil / 2 + 1 }}{4 \cdot \omega}, there will be no non-served clients. Therefore \mra{} serves all the clients before it finishes.

Now we bound the sum of latencies produced by \mra{} on the clients. At each \stp{q}{f}, we add a path for each repairman $r_i$ who can travel its path and come back in time at most $2\cdot q \cdot \mu$. Therefore repairman $r_i$ travels at most $8 \cdot \omega \cdot q \cdot \mu$ units of time for the paths added at \stp{q}{f} for $1 \leq f \leq 4\cdot \omega$. Repairman $r_i$ before starting to travel the path added at \stp{q}{1} has to travel all the paths added before. Traveling previous paths takes the total of $\sum_{i = 0}^{\log q - 1} 8 \cdot \omega \cdot 2^i \cdot \mu \leq 8 \cdot \omega \cdot \mu \cdot q$. Therefore we can assume that all the clients that are served at \stp{q}{f} for $1 \leq f \leq 4\cdot \omega$ have latency at most $16 \cdot \mu \cdot \omega \cdot q$ (Fact 1), $8 \cdot \omega \cdot q \cdot \mu$ for the paths added at $\stp{q}{.}$ and $8 \cdot \omega \cdot q \cdot \mu$ for the paths added before \stp{q}{1}.

From Claim \ref{claim:stp} we know for $q = 2^a$ we have $F^{q, \fw + 1} \leq \sum_{0 \leq s \leq a} \left(\frac{1}{4^{a-s+1}}\right) h_{(2^s)}$. Remember that set $\A^{q, 4 \cdot \omega}$ is the set of all the clients that \mra{} has served in or before executing \stp{q}{4 \cdot \omega}. From the definition of $\A^{q, 4 \cdot \omega}$ we conclude the following inequality.
\begin{align}
\notag |\A^{q, 4 \cdot \omega}| &\geq \sum_{0 \leq s \leq a} h_{(2^s)} - \sum_{0 \leq s \leq a} \left(\frac{1}{4^{a-s+1}}\right) h_{(2^s)} & \text{ \cref{setaa}}\\
\label{frac-cover} &= \sum_{0 \leq s \leq a} \left(\frac{4^{a-s+1} - 1}{4^{a-s+1}}\right) h_{(2^s)}
\end{align}

Note that the size of $\A^{q, 4 \cdot \omega}$ is integral. In fact, Inequality (\ref{frac-cover}) shows \mra{} serves at least $\left\lceil \sum_{0 \leq s \leq a} \left(\frac{4^{a-s+1} - 1}{4^{a-s+1}}\right) h_{(2^s)} \right\rceil$ clients after it executes \stp{q}{\fw}. Note that from Fact 1 each client that is served in or before execution of \stp{q}{\fw} sees a latency at most $16 \cdot \mu \cdot \omega \cdot q$ and each client that is served in the succeeding execution of \stpn{} sees a higher latency. For the sake of explanation and to avoid dealing with the ceiling function, we slightly abuse the notation $|\A^{q, 4 \cdot \omega}|$. We assume that \mra{}  serves exactly $\sum_{0 \leq s \leq a} \left(\frac{4^{a-s+1} - 1}{4^{a-s+1}}\right) h_{(2^s)}$ (possibly fractional) clients after it executes \stp{q}{\fw} and assume that it serves the rest (if there are more) in the succeeding executions of \stpn{}. Note that this way we do not decrease the total latency.
We use the following claim to upper bound the total latency of \mra{}.
\begin{nclaim}
\label{claim:latency}
After \mra{} executes \stp{q}{4 \cdot \omega}, the total latency of $\sum_{0 \leq s \leq a} \left(\frac{4^{a-s+1} - 1}{4^{a-s+1}}\right) h_{(2^s)}$ clients that are served is at most $\sum_{0 \leq s \leq a} 32 \cdot \mu \cdot \omega \cdot 2^s \cdot \left(\frac{4^{a-s+1} - 1}{4^{a-s+1}}\right) h_{(2^s)}$.
\end{nclaim}
\begin{proof}
Let $q = 2^a$, we prove the claim by induction on $a$. For the base case when $a = 0$ Inequality (\ref{frac-cover}) implies \mra{} serves at least $\frac{3}{4} h_1$ clients after it executes \stp{1}{\fw}. From Fact 1 the total latency for these clients is $16 \cdot 2^0 \cdot \omega \cdot \mu \cdot \frac{3}{4} h_0$.

For $a = k'$ we assume after \mra{} executes \stp{2^{k'}}{\fw} it serves $\sum_{0 \leq s \leq k'} \left(\frac{4^{k'-s+1} - 1}{4^{k'-s+1}}\right) h_{(2^s)}$ clients with total latency $\sum_{0 \leq s \leq k'} 32 \cdot \mu \cdot \omega \cdot 2^{s} \cdot \left(\frac{4^{k'-s+1} - 1}{4^{k'-s+1}}\right) h_{(2^s)}$ as the induction hypothesis.

For $a = k' + 1$, Inequality (\ref{frac-cover}) shows \mra{} serves at least $\sum_{0 \leq s \leq k' + 1} \left(\frac{4^{k' + 1-s+1} - 1}{4^{k' + 1-s+1}}\right) h_{(2^s)}$ clients after it executes \stp{2^{k' + 1}}{\fw}. From the induction hypothesis we serve $\sum_{0 \leq s \leq k'} \left(\frac{4^{k'-s+1} - 1}{4^{k'-s+1}}\right) h_{(2^s)}$ clients with latency $\sum_{0 \leq s \leq k'} 32 \cdot \mu \cdot \omega \cdot 2^s \cdot \left(\frac{4^{k'-s+1} - 1}{4^{k'-s+1}}\right) h_{(2^s)}$. Each of the remaining $\sum_{0 \leq s \leq k' + 1} \left(\frac{4^{k' + 1-s+1} - 1}{4^{k' + 1-s+1}}\right) h_{(2^s)} - \sum_{0 \leq s \leq k'} \left(\frac{4^{k'-s+1} - 1}{4^{k'-s+1}}\right) h_{(2^s)}$ clients are served with latency $16 \cdot \mu \cdot \omega \cdot 2^{k' + 1}$ (from Fact 1). We bound the total latency in the following inequalities which finish the proof of the induction and hence the claim.
\begin{align*}
&\sum_{0 \leq s \leq k'} 32 \cdot \mu \cdot \omega \cdot 2^s \left(\frac{4^{k'-s+1} - 1}{4^{k'-s+1}}\right) h_{(2^s)} + \\
&\left(\sum_{0 \leq s \leq k' + 1} \left(\frac{4^{k' + 1-s+1} - 1}{4^{k' + 1-s+1}}\right) h_{(2^s)} - \sum_{0 \leq s \leq k'} \left(\frac{4^{k'-s+1} - 1}{4^{k'-s+1}}\right) h_{(2^s)}\right) \cdot 16 \cdot \mu \cdot \omega \cdot 2^{k' + 1}\\
&= \sum_{0 \leq s \leq k'} 32 \cdot \mu \cdot \omega \cdot 2^s \cdot \left(\frac{4^{k'-s+1} - 1}{4^{k'-s+1}}\right) h_{(2^s)} + \\
&\left(\sum_{0 \leq s \leq k'} \left(\frac{3}{4^{k' + 1-s+1}}\right) h_{(2^s)}\right) \cdot 16 \cdot \mu \cdot \omega \cdot 2^{k' + 1} + \left(\frac{3}{4}\right) h_{\left(2^{k' + 1}\right)} \cdot 16 \cdot \mu \cdot \omega \cdot 2^{k' + 1}\\
&= \sum_{0 \leq s \leq k'} 32 \cdot \mu \cdot \omega \cdot 2^s \cdot \left(\frac{4^{k' + 1-s+1} - 1}{4^{k' + 1-s+1}}\right) h_{(2^s)}  + \left(\frac{3}{4}\right) h_{\left(2^{k' + 1}\right)} \cdot 16 \cdot \mu \cdot \omega \cdot 2^{k' + 1}\\
&\leq \sum_{0 \leq s \leq k' + 1} 32 \cdot \mu \cdot \omega \cdot 2^s \cdot \left(\frac{4^{k' + 1-s+1} - 1}{4^{k' + 1-s+1}}\right) h_{(2^s)}
\end{align*}
\qed
\end{proof}

Note that from \cref{def-h} for $h_q$, we can rewrite the objective value of $(x,y)$ in \ref{rplp} as $\sum_{0 \leq s \leq \left\lceil \log T \right\rceil + \left\lceil \log m \right\rceil / 2  + 1} 2^s h_{(2^s)}$. The total latency of \mra{} is at most $\sum_{0 \leq s \leq \left\lceil \log T \right\rceil + \left\lceil \log m \right\rceil / 2  + 1} 32 \cdot \mu \cdot \omega \cdot 2^s \cdot h_{(2^s)}$ from Claim \ref{claim:latency}. Therefore \mra{} has the total latency at most $32 \cdot \mu \cdot \omega \cdot \opt $.
\qed
\end{proof}

\input{npcst.tex}

\input{euclidean.tex}
\input{maxmr.tex}

\section{Acknowledgment}
The authors would like to thank an anonymous reviewer of Approx+Random 2013 conference who provided us with helpful comments and more importantly an alternative proof which gives an $\left( O(1), O(1), O(\log n)\right)$-approximation algorithm for the NPCST problem.

\bibliographystyle{alpha}
\bibliography{refs}


\end{document}

%% file: appendixintors.tex
\section{Connection of MR to the Other Class of Problems}
\label{connections}
Chakrabarty and Swamy \cite{CS10} define a general problem called Minimum Latency Uncapacitated Facility Location (MLUFL). In MLUFL we are given a set $F$ of $n$ facilities with opening costs $\{f_i\}$, a set $D$ of $m$ clients, a root node $r$, and connection costs $\{c_{ij}\}$ for connecting client $j$ to facility $i$. The objective is to select a subset $F'$ of facilities to open, find a path $p$ staring from $r$ to visit all the facilities in $F'$, and assign each client $j$ to a facility $\Phi(j)$, to minimize $\sum_{i \in F'} f_i + \sum_{j \in D} \left(c_{\Phi(j)j} + t_j\right)$ where $t_j$ is the distance of $\phi(j)$ from $r$ in the path $p$. In the {\em related} MLUFL problem the  connection cost $c_{\Phi(j)j}$ is the distance of client $j$ to facility $\Phi(j)$ in the metric space. Chakrabarty and Swamy \cite{CS10} provide a constant factor approximation algorithm for the related MLUFL problem. They also generalize it to the case when instead of one activating path $p$ we can have $k$ activating paths. This result gives a constant factor approximation algorithm for the special case of \smr{} when all the repairmen have the same speed and start from the same depot. In \smr{} we want to minimize $\sum_{j \in D} \max\left(c_{\Phi(j)j}, t_j\right)$ as opposed to $\sum_{j \in D} \left(c_{\Phi(j)j} + t_j\right)$ in related MLUFL where the facility opening costs are zero; but note that they have a multiplicative difference of at most $2$.

In the rest of this section we show how the MR problem is related to the other well-studied problems in computer science. 
In particular we show its connections to the movement framework, neighborhood TSP problems, and Orienteering problems.
\subsection{Connection to the movement framework}
\label{intro-movement}
\smr{} can be defined under the {\em movement} framework first introduced by Demaine et al.~\cite{DHM+09}. In the movement framework, we are given a general weighted graph and pebbles with different colors are placed on the nodes of the graph forming a starting configuration. The goal is to move the pebbles such that the final configuration of the pebbles meets a given set of properties. We can think of the repairmen as blue pebbles and the clients as red pebbles on the metric completion of the graph. Each red or blue pebble can move with different speeds. The latency of a red pebble is the earliest time it is collocated with a blue pebble. The objective in \smr{} is to minimize the total latency.

The paper by Demaine et al.~\cite{DHM+09} inspired several other
papers on the movement problems and there are a handful of recent approximation algorithms for them as well. Of primary relevance to this paper, is the work of
Friggstad and Salavatipour~\cite{friggstad-salavatipour-focs} who
consider minimizing movement in the facility location setting in
which both facilities and clients are mobile and can move. The
quality of a solution can be measured either by the total distance
(average distance) clients and facilities travel or by the maximum
distance traveled by any client or facility. They obtain constant-factor approximation algorithms for these problems and recently, Friggstad \etal \cite{AFS13} improve the constant factor to $(3 + \epsilon)$ for a slightly more general problem. Very recently,
Berman, Demaine, and Zadimoghaddam~\cite{BDZ11} obtain a constant-factor approximation algorithm for minimizing maximum movement to
reach a configuration in which the pebbles form a connected
subgraph. This result is very interesting since Demaine et
al.~\cite{DHM+09} show with the total sum movement object function,
the connectivity movement problem is $\Omega(n^{1-\epsilon})$
inapproximable, and for which there is only an
$\tilde{O}(n)$-approximation algorithm. Demaine, Hajiaghayi,
and Marx~\cite{DHM09} consider the problem when we have a relatively
small number of mobile agents (e.g., a team of autonomous robots,
people, or vehicles) moving cooperatively in a vast terrain or
complex building to achieve some task. They find optimal solutions for several movement problems in polynomial time where the number of pebbles is constant.  Finally there are various
specific problems considered less formally in practical scenarios
\cite{rus:icra04,rus:iser04,habfm-ardrue-03,LaValle06,ReifW95,MRS03,DMMMZ97,JBQZ04,StrijkW01,JQQZC03}
as well.

\subsection{Connection to the Neighborhood TSP problem}
\label{sec-intro-ntsp}
Suppose we guess in advance that a client $c$ is going to be served by time $t_c$. \footnote{We are able to estimate time $t_c$ for each client with the help of our LP for \smr{}.} Client $c$ can go to any node which is reachable from its starting location by time $t_c$ in order to be served. The set of such reachable nodes defines a neighborhood for each client. On the other hand, repairmen's task is to visit a node in each client's neighborhood in order to serve him. In fact, giving the clients ability to move in \smr{} can be thought as assigning a neighborhood to each client.

Neighborhood problems are well-studied in the theoretical computer science community.
A well-studied problem in this category is the Neighborhood TSP (NTSP) problem where we are given a set of clients each with a neighborhood and the goal is to find a minimum tour to visit at least one node from the neighborhood of each client. Visiting a neighborhood instead of a node makes the problem considerably harder; interestingly we will show an $\Omega(\log^{2 - \epsilon})$-hardness for the NTSP. This hardness is indeed the source of subtlety in \smr{}. 

NTSP is especially studied for the Euclidean Space in the CS community, as it has routing-related and VLSI design applications \cite[Chapter~15]{M00} \cite{RW89}. Euclidean Neighborhood TSP (ENTSP) is known to be APX-hard \cite{BGK+05,SS06,EFS06}. A neighborhood is called fat if we can fit a disc inside the neighborhood such that the radius of the disc is at least a constant fraction of the radius of the neighborhood, where the radius of a neighborhood is half of the distance between its two farthest points. A PTAS for ENTSP is known, when the neighborhoods are fat, roughly the same size (their sizes differ in at most a constant factor), and overlap with each other in at most a constant number of times\cite{DM01,FG04,M07}. When neighborhoods are connected and disjoint Mitchell \cite{M10} gives a constant-factor approximation algorithm for the problem. Elbassioni \etal \cite{EFS06} give a constant-factor approximation for ENTSP when the neighborhoods are roughly the same size, convex, fat, and can intersect.

In this paper we generalize ENTSP to the prize-collecting version. In the prize collecting ENTSP we are given a budget $B$ and each neighborhood is assigned a profit, the objective is to find a tour of maximum length $B$ to maximize the sum of profits of the neighborhoods the tour intersects.  We give a constant-factor approximation algorithm for the prize-collecting ENTSP when the neighborhoods are fat and roughly the same size and can intersect. As a result of this algorithm, we obtain a constant factor approximation algorithm to \smr{} in Euclidean Space.

\subsection{Connection to the Orienteering Problem}
\label{intro-orienteering}
It turns out that the separation oracle for the dual of our LP is closely related to the {\em orienteering problem}. In the orienteering problem we are given a graph $G = (V, E)$, two nodes $s,t \in V$ and a length bound $B$ and the goal is to find a tour starting from $s$ and ends in $t$ with length at most $B$ that visits maximum number of nodes. Orienteering is shown to be NP-Hard via an easy reduction from the TSP problem and it is also APX-hard \cite{BCK+03}. Blum \etal \cite{BCK+03} give the first constant-factor approximation algorithm with ratio $4$ for orienteering which is improved to $3$ by Bansal \etal \cite{BBC+04}. Chekuri \etal \cite{CKP08} give a ($2 + \epsilon$)-approximation algorithm for the problem in the undirected graphs and an $O(\log^2 \opt)$-approximation algorithm in directed graphs, which are the current best approximation factors for orienteering.


%% file: results.tex
\section{Results and techniques}
\label{sec-results}
In this section we summarize all our results along with the overview of their proofs. All the ideas explained here are new in this context.

Our main result is an $O(\log n)$-approximation algorithm for \smr{}. More precisely we prove the following theorem.
\begin{theorem}
\label{const-theorem}
There is an $O(\log n)$-approximation algorithm for the \smr{} problem which also upper bounds the integrality gap of its LP formulation.
\end{theorem}

We present the novel properties (in this context) of our techniques.
First we relax conditions on serving the clients. If a client
collocates with a repairman at a certain node of the metric space
during the movements we say it gets served {\em perfectly}. On the
other hand if a client visits a node through which a repairman has
passed no later than the arrival of the client, we say it gets served
{\em indirectly}. We design a procedure that transforms any solution
to the \smr{} problem where the clients are served indirectly to a
solution where all the clients are served perfectly by increasing the
total latency with a multiplicative factor at most $3 + \epsilon$. We
solve \smr{} for the case when we serve the clients indirectly and use
the procedure to serve the clients perfectly.  We give an LP
formulation for \smr{} (to serve the clients indirectly) and bound its
integrality gap. However there are two major challenges in order to do
so. First, solving the LP which has exponentially many variables and
second, rounding a solution to the LP efficiently to an integral
solution.

In order to solve the LP we need a separation oracle for its dual which turns out to be the following problem.
\begin{definition}
\label{npcst-def}
Neighborhood Prize Collecting Steiner Tree (NPCST): An instance of the NPCST problem consist of an ordered tuple $(V, d, r, C, L)$ where $V$ is the set of nodes, $d$ is a metric distance function on the set $V$, $r \in V$ is the root node, $C$ is the set of clients, and $L$ is the cost budget. Each client $c \in C$ is associated with a profit $\theta_c$ and a neighborhood ball $\B(c, t_c)$ which contains all the nodes $u$ with $d(u,c) \leq t_c$. The goal is to find a Steiner tree $T_{\OPT}$ such that $cost(T_{\OPT}) \leq L$ and the sum of the profits of the clients whose $\B$-ball hits $T_{\OPT}$ is maximized.
\end{definition}
The vehicle routing problems become significantly harder when instead of visiting a node,  it is sufficient to visit a neighborhood around it. For example in the Neighborhood Steiner Tree (NST) problem, we are given a graph $G$ with a set of clients $C$ where each client $c$ is associated with a neighborhood ball. The objective for NST is to find a tree $T$ with minimum weight that serves at least one node from each client's neighborhood ball. We will prove the following hardness result about the NST problem which shows the source of difficulty in our problem.
\begin{theorem}
\label{npcst-hardness}
There is no $O(\log^{2 - \epsilon} n)$-approximation algorithm for the NST problem unless NP has quasi-polynomial Las-Vegas algorithms.
\end{theorem}

To avoid this hardness we allow relaxing the NPCST constraints. More formally we accept a tri-criteria approximation algorithm for NPCST as our separation oracle defined formally below.
\begin{definition}
A $(\sigma, \phi, \omega)$-approximation algorithm for the instance $(V, d, r, C, L)$ of the NPCST problem finds a Steiner tree $T$ with the following properties; $T$ is said to hit a client $c$ with $\B$-ball $\B(c,t_c)$ if $T$ has at least one node in $\B(c, t_c \cdot \sigma)$, the weight of $T$ is at most $\phi \cdot L$, and sum of the profits of the clients got hit by $T$ is at least $\frac{1}{\omega} \opt$ where $\opt$ is the amount of profit an optimum tree collects with no violation in any bound.
\end{definition}

Accepting a tri-criteria approximation algorithm for NPCST has two benefits. Firstly, it reduces the difficulty of solving the NPCST problem to avoid the hardness results similar to \cref{npcst-hardness}. Secondly, later when we transform the solution of the algorithm to  a solution of \smr{}, it allows the approximation factor on the traveling time for a client and a repairman to reach to a certain node (latency) to get split  between both the client (violating its neighborhood) and the repairman (violating the weight of the tree). However a solution to NPCST resulting from a tri-criteria approximation algorithm is harder to transform to a solution of \smr{}.
We prove the following general theorem to transform any tri-criteria approximation algorithm to the NPCST problem to an efficient approximation algorithm for \smr{}.
\begin{theorem}
\label{general-thm}
Given a $(\sigma, \phi, \omega)$-approximation algorithm for NPCST, we can find an $O(\max(\sigma, 2\phi) \cdot \omega)$-approximation algorithm for \smr{}.
\end{theorem}

Proving \cref{general-thm} has two parts. The first part is to use the tri-criteria approximation algorithm to find a feasible solution for our LP and the second part is to round the feasible solution. For the first part we introduce a new relaxed LP for \smr{} to absorb the violations of the tri-criteria approximation algorithm while keeping the optimal value of the relaxed LP to be at most the optimal value of the original LP. Then we show that using the $(\sigma, \phi, \omega)$-approximation algorithm for NPCST, we can find a feasible solution to the relaxed LP with the objective value at most the optimum value of the original LP.

For the second part of the proof, we round the feasible solution found in the previous part to an integral solution for \smr{} with the total latency at most $O(\max(\sigma, 2\phi) \cdot \omega)$ times the optimal value of our LP.
Our algorithm (later given in \cref{mra}) for the rounding part is easy to state and implement but needs a relatively complicated analysis.
The algorithm runs in several steps where each step represents a time-stamp. The time-stamps increase geometrically, \ie the time-stamp of a step is twice as the time-stamp of the previous step. At each step we randomly select a tour for each repairman from the set of all tours with the length at most the time-stamp times the repairman's speed. The random selection is done using the LP values. The output of our algorithm is the concatenation of all the tours selected at each step.
The idea for the analysis of our algorithm is as follows. Let the time-stamp for a certain step be $2^a$ and $F$ be the (fractional) number of clients that are served by time $2^a$ according to the LP values. We show the expected number of clients that our algorithm serves in the step is at least $\frac{3F}{4}$. We show that this condition is enough to bound the total latency of the clients. Finally, we show that our algorithm can be derandomized. The derandomization is done by a recursive algorithm which takes an arbitrary subset ($R'$) of $R$ and selects a path for a repairman $r$ in $R'$ and calls itself with parameter $R' \setminus \{r\}$. It selects a path for $r$ which covers the maximum number of clients from the set of clients that are served fractionally by repairmen of $R'$ in the LP solution but not served by the paths we have selected till now. We prove that the greedy algorithm serves at least $\lceil \frac{3F}{4} \rceil$ clients by induction on the size of set $R'$.

We prove the following theorem about the NPCST problem for the general metrics which is of independent interest and non-trivial. In order to prove \cref{const-theorem}, we plug this result about the NPCST problem  in \cref{general-thm}. 
\begin{theorem}
\label{npcst-thm}
There is an $\left( O(\log n), O(\log n), 2 \right)$-approximation algorithm for the NPCST problem in general metrics.
\end{theorem}
Remember that in the NPCST problem we have to find a tree $T$ to maximize the number of clients whose $\B$-ball contain a node of $T$. To prove the above theorem we embed the graph into a distribution of tree metrics\cite{FRT04}. Note that the $\B$-balls for the clients are not preserved in the tree metrics. We define a new problem on the tree metrics as follows. Given a budget $L'$ we want to find a tree $T$ with weight at most $L$ to maximize the size of the set of served clients $C'$ where $T$ serves set $C'$ if the sum of distances of the clients in $C'$ to $T$ is at most $L'$. We solve the new problem efficiently in the tree metrics with dynamic programming. Finally we show by violating the radii of the $\B$-balls of the clients by a constant factor, $T$ actually serves a good fraction of the clients in $C'$  (by Markov's Inequality)  through hitting their $\B$-balls.

An anonymous referee pointed out that we can also obtain an $\left(O(1), O(1), O(\log n) \right)$-approximation algorithm for the NPCST problem in general metrics using the ideas in \cite{GKKRY01} and \cite{SK04}. This algorithm is interesting since as opposed to the algorithm of \cref{npcst-thm} there is no violation on the cost of the tree and the neighborhoods' radii but it collects an $O(\log n)$ fraction of the optimal profit. Note that by plugging this algorithm into \cref{general-thm} we get the same result for the NPCST problem as in \cref{const-theorem}. The description of the algorithm and an outline of the proof given by the referee is brought in \cref{ref-alg}. 

Motivated from the application of \smr{} in Amazon Locker Delivery \cite{amazon} and USPS gopost \cite{gopost} which occurs in the Euclidean space, we prove the following theorem to get a constant-factor approximation algorithm for \smr{} in the Euclidean space. The neighborhood problems are also especially studied in the geometric settings. The usual assumption in the neighborhood TSP problems  for getting a constant factor approximation (see \cref{sec-intro-ntsp}) is to assume the radius of the biggest neighborhood is at most a constant factor larger than the smallest one. We plug the following theorem in \cref{general-thm} to get a constant factor approximation algorithm for \smr{}. Here the radius constraints means the maximum speed of the clients is at most a constant factor larger than the minimum speed which is an acceptable constraint for the package delivery problem motivating \smr{}.

\begin{theorem}
\label{npcst-2d}
There is an $(O(P), O(1), O(1))$-approximation algorithm for the NPCST problem in the Euclidean space where the radius of the greatest neighborhood is at most $P$ times larger than the radius of the smallest neighborhood.
\end{theorem}

Our last result is a constant-factor approximation algorithm for Max-MR.
\begin{theorem}
\label{thm:max-mr}
There is a constant-factor approximation algorithm for the Max-MR problem when the repairmen have the same speed.
\end{theorem}

%% file: npcst.tex
\section{Neighborhood Prize Collecting Steiner Tree}
\label{npcst-sec}
We start this section by proving \cref{npcst-hardness}.
\subsection{Proof of \cref{npcst-hardness}}
\begin{proof}
\label{missingproofs-npcst-hardness}
As described in the introduction section, serving neighborhoods instead of a single node in the covering problems like neighborhood TSP makes the problem significantly harder. We show an $\Omega(\log^{2 - \epsilon})$ hardness result about Neighborhood Steiner Tree (NST) problem defined formally below.
\begin{definition}
An instance $\I = (\M, C)$ of NST consists of a metric space $\M = (V, d)$ where $V$ is the node set and $d:(V \times V) \rightarrow \mathbb{Q}^+$ is the distance function, and a set of clients $C$ where each client $c$ has a neighborhood ball $\B(c,t_c)$ which is the set of all nodes $u$ such that $d(u, c) \leq t_c$. The objective for NST is to find a tree $T$ with minimum cost such that for each client $c \in C$, $T$ serves at least one node of $\B(c,t_c)$.
\end{definition}
Note that as opposed to NPCST, in NST we do not have bound on the resulting tree and profits for the clients but we have to serve all the clients.

We reduce from the Group Steiner Tree (GST) problem. An instance $\I' = (\M', \G')$ of GST consists of a metric space $\M' = (V', d')$ where $V'$ is the node set and $d':(V' \times V') \rightarrow \mathbb{Q}^+$ is the distance function and a collection of groups $\G' = \{G'_1, G'_2, \ldots, G'_k\}$ where each group $G'_i$ is a subset of nodes in $V'$. The objective for GST is to find a tree $T'$ with minimum cost such that serves at least one node from each group.  We use the following result of Halperin and Krauthgamer \cite{HK03} about the group Steiner tree problem.

\begin{theorem}
\label{gst-hardness}
\cite{HK03}: For every fixed $\epsilon > 0$, group Steiner tree cannot be approximated within ratio $\log^{2-\epsilon} k$, unless $NP \subseteq ZTIME(n^{polylog(n)})$; this holds even for trees.
\end{theorem}

We show a polynomial-time reduction from any instance $\I'$ of GST to an instance $\I$ of NST where any $\alpha$-approximation for $\I$ results in an $\alpha$-approximation algorithm for $\I'$. We build $\I$ from $\I'$ as follows. Let assume the distance between the farthest pair of points in $V'$ is $M'$. The node set in $\I$ is $V'$ plus $k$ dummy nodes $y_1, y_2, \ldots, y_k$ where each $y_i$ corresponds to $g'_i$ in $\G'$, \ie $V = V' \cup \{y_1, y_2, \ldots, y_k\}$. The distances between nodes of $V$ that correspond to the nodes in $V'$ is the same as $d'$. Each dummy node $y_i$ is connected to all the nodes $u$ whose corresponding node in $V'$ is in $G'_i$ with an edge of length $M'$. The set of clients $C$ in $\I$ is $\{y_1, \ldots, y_k\}$ where each $y_i$ has neighborhood ball with radius $M'$ ($\B(y_i, M')$).

Note that  this reduction can be done in polynomial time. Any solution $T'$ for GST in $\I'$ by taking the corresponding nodes and edges in $\I$ can be transformed to a solution $T$ for NST in $\I$ with the same cost since $T'$ serves at least one client from each group in $\I'$ and hence $T$ serves one node from each neighborhood in $\I$. Now we prove that any solution $T$ for NST in $\I$ can be transformed to a solution $T'$ for GST in $\I'$ with at most the same cost. $T'$ contains each node $v'_i$ whose corresponding node $v_i$ is in $T$ and is not a dummy node. Firstly, if a dummy node $y_i$ is a leaf in $T$ we simply ignore $y_i$ in $T'$ without affecting servicing the groups since $y_i$ is not a node of any group in $\G'$. Secondly, if $T$ uses a dummy node to go from node $v_i$ to a node $v_j$ which costs $2M$, we instead use the shortest path between $v'_i$ and $v'_j$ (the corresponding nodes to $v_i, v_j$) in $V'$ which costs at most $M$. Thus, $T$ can be transformed to a tree $T'$ in $\I'$ using only nodes in $V'$ that serves all the groups, with at most the same cost as $T$.

Because each solution for $\I$ can be transformed to a solution to $\I'$ with at most the same cost and vice versa, we can conclude that the optimum solution for both $\I'$ and $\I$ have the same cost. Moreover, we can get the output of an $\alpha$-approximation algorithm for $\I$ and transform it to a solution to the GST with at most the same cost which is an $\alpha$-approximation algorithm for the GST problem. Thus, from \cref{gst-hardness} we conclude that there is no $O(\log^{2-\epsilon} k)$-approximation algorithm for NST, unless $NP \subseteq ZTIME(n^{polylog(n)})$ for every fixed $\epsilon > 0$.
\qed
\end{proof}

\subsection{Proof of \cref{npcst-thm}}

We use the result of Fakcharoenphol \etal \cite{FRT04} to embed metric space  $\M = (V, d)$ into a distribution of dominating trees with a distortion at most $O(\log n)$ \footnote{The Metric embedding problems are studied well in the theory community. Bartal \cite{B96} first defined probabilistic embeddings and gives the distortion ratio $O(\log^2 n)$. Bartal \cite{B98} improved this ratio to $O(\log n \log \log n)$ by using the ideas inspired from Seymour's work on feedback arc set \cite{S95}. The $O(\log n)$ distortion ratio of Fakcharoenphol \etal \cite{FRT04} is the best possible distortion ratio one can hope for.}. Here, a tree $T$ is dominating $d$ if for any two nodes $u,v \in V$ we have $d_T(u,v) \geq d(u,v)$ where $d_T(u,v)$ is the length of the unique path between $u$ and $v$ in $T$. Fakcharoenphol \etal show that their embedding can be done deterministically into at most $O(n \log n)$ dominating trees where $n$ is the number of the nodes using the same technique first introduced by Charikar \etal \cite{CCG+98}.  More formally, they proved that the metric $d$ can be embedded into a distribution $\pi$ of  $p$ trees $T_1, T_2, \ldots, T_p$ where $p \in O(n \log n)$ such that the following holds.
\begin{equation*}
\forall u,v \in V,\quad  \sum_{i = 1}^p \pi(i) d_{T_i}(u,v) < O(\log n) \cdot d(u,v)
\end{equation*}
In order to be more specific we assume that $O(\log n)$ in the above inequality is $A \log n$ for an appropriate constant value $A$, more formally we use the following inequality.
\begin{equation}
\label{frt}
\forall u,v \in V,\quad  \sum_{i = 1}^p \pi(i) d_{T_i}(u,v) < A \cdot \log n \cdot d(u,v)
\end{equation}
Note that the average distortion in the distances of all the pairs are bounded, however we cannot ensure that there is a single tree in which all the distances are distorted with at most a factor of $ A \cdot \log n$. Moreover, the neighborhood balls of the clients are also not preserved in the trees. Therefore, instead of hitting the neighborhoods of the clients we try to minimize the total service cost of all the clients where the service cost of a client ($c$) in a tree ($H$) is its profit $\theta_c$ times its distance to $H$ divided by $t_c$ (the radius of its neighborhood ball). More formally, We define the following alternative problem to solve for each tree $T_i$ which helps us solve the NPCST problem in the original metric $d$.

\begin{definition}
\label{def-tscst}
Instance $\I = (T = (V,E), r, C, B, B')$ of the Total Service Cost Steiner Tree (TSCST) problem consist of a tree $T$ on the node set $V$ and edges $E$ rooted at $r$ where each client $c \in C$ has profit $\theta_c$ and is assigned to a neighborhood ball with radius $t_c$, length bound $B \in \bbZ^+$, and service cost $B' \in \bbQ^+$. The objective is to find a tree $H$ (subtree of $T$) to serve a subset $C' \subseteq C$ such that $\sum_{c \in C'} \theta_v$ is maximized. The constraints on $H$ are as follow: (1) the total length of $H$ has to be at most $B$ and (2) the service cost of $C'$ which is defined as $\sum_{c \in C'} \theta_c \cdot \frac{d_T(c, H)}{t_c}$ has to be at most $B'$ where $d_T(c, H)$ is the distance of the location of $c$ to its nearest node in $H$ in metric $d_T$.
\end{definition}
Intuitively, in the TSCST as opposed to NPCST, we try to minimize the overall violations on the radius of the neighborhoods instead of having hard capacity on the neighborhoods' radius. 
\begin{lemma}
\label{dp-gen}
For an arbitrary small positive value $\epsilon > 0$, we can design an algorithm whose running time is polynomial in $n$ and $\frac{1}{\epsilon}$ which solves the TSCST problem efficiently while violates service cost bound $B'$ by a factor at most $(1 + \epsilon)$.
\end{lemma}
Note that the profits in the NPCST problem ($\theta$) come from the dual variables of our LP.  We are interested in the service costs that are integers and polynomially bounded by the size of input in order to design an efficient algorithm to solve TSCST. Therefore, we define value $X$ to be $\frac{B' \cdot \epsilon}{|C|}$ and scale and round the service cost for each client $c$ to value $\floor{\frac{\theta_c \cdot d_T(v, H)}{t_c \cdot X}}$. More formally, we define the following scaled version of the TSCST problem in which all the service costs are positive integers as opposed to rational numbers.

\begin{definition}
\label{STSCST}
Instance $\hat{\I} = (T = (V,E),r, C, B, \hbp)$ of Scaled Total Service Cost Steiner Tree (STSCST) problem consist of tree $T$ rooted at $r$ on the node set $V$ and edges $E$, set of clients $C$ where each client $c \in C$ has profit $\theta_c \in \bbQ^+$, length bound $B \in \bbZ^+$, and service cost bound $\hbp \in \bbZ^+$. The objective is to find a tree (subtree of $T$) ($H$) rooted at $r$ to serve a subset $C' \subseteq C$ such that $\sum_{c \in C'} \theta_c$ is maximized. The constraints on $H$ are as follow: (1) the total length of $H$ has to be at most $B$ and (2) the service cost of $C'$ which is defined as $\sum_{c \in C'} \floor{\frac{\theta_c \cdot d_T(c, H)}{t_c \cdot X}}$ has to be at most $\hbp$ where $d_T(c, H)$ is the distance of location of $c$ to its nearest node in $H$ in metric $d_T$.
\end{definition}
In order to prove \cref{dp-gen}, for a given instance $\I = (T = (V,E),r, C, B, \hbp)$ of TSCST we create a corresponding instance ($\hat{\I} = (T = (V,E),r, C, B, \hbp$) of STSCST where $\hbp = \floor{\frac{B}{X}}$.
Clearly every solution to $\I$ is a solution to $\I'$ as we round down all the service costs. Therefore, the optimum value of $\I'$ is at least the optimum value $\I$. In the following lemma we show that STSCST can be solved efficiently. We find an optimal solution to $\I'$ and announce it a as a solution to TSCST, however the solution might violate the service cost bound $\hbp$. The violation is at most $X \cdot |C|$ since we have at most $|C|$ clients and for each client the scaled value of service cost can be $X$. As $X \cdot |C|$ is equal to $B' \cdot \epsilon$, the solution we find in $\I'$ collects the maximum profit and violate the service cost by at most a factor of $\epsilon$ in $\I$ and hence the proof of \cref{dp-gen} follows.

\begin{lemma}
\label{dp}
There is an algorithm with running time $O(n\cdot B^2 \cdot \hbp^2 + n \cdot |C| \cdot B \cdot \hbp^3)$ and solves the STSCST problem efficiently.
\end{lemma}
\begin{proof}
Note that since we assume in the input of NPCST all the distances are polynomially bounded, value $B$ is also polynomially bounded. Value $\hbp$ is at most $|C|\frac{1}{\epsilon}$. Therefore, the running time of our algorithm is polynomial in terms of inputs.

We replace every node with at least $3$ children by a complete binary tree of edges of cost zero (except for the last edge
entering a leaf that has the weight) thus imposing at most $2$ children for any node. Therefore, we assume that every node in $T$ has at most two children. The {\em level number} of a node $v$ in $T$ is the number of edges in the unique path from $v$ to $r$. We denote by $sub(v)$ the subtree containing node $v$ as its root and all its children.

We design a Dynamic Programming (DP) to solve the problem. The subproblems to our DP algorithm are defined for the subtrees rooted at each node and all possible combination of length bounds and service costs.
\begin{definition}
\label{def-sub}
For each node $v$, value $B_v \left(0 \leq B_v \leq B\right)$, and value $\hbp_v (0 \leq \hbp_v \leq \hbp)$, sub-problem ($DP(v, B_v, \hbp_v)$) to our dynamic programming is instance $\left(sub(v), r, C_v, B_,\hbp_v\right)$ of STSCST where $C_v$ is the set of all clients whose starting locations are in $sub(v)$.
\end{definition}

We call that sub-problem $DP(u, B_u, \hbp_u)$ is smaller than $DP(v, B_v, \hbp_v)$ if $u$ is a children of $v$. The following claim shows that each sub-problem $DP(v, B_v, \hbp_v)$ can be solved efficiently using only smaller sub-problems.

\begin{nclaim}
\label{dp-sub}
Each subproblem $DP(v, B_v, \hbp_v)$ can be solved in polynomial time given we can find an optimum solution for all its smaller sub-problems.
\end{nclaim}
\begin{proof}
In order to solve $DP(v, B_v, \hbp_v)$ we need to use an algorithm which solves the Knapsack problem. In the Knapsack problem we are given a set of items $i_1, i_2, \ldots, i_q$ where each item $i_j$ has weight $w(i_j) \in \bbZ^+$ and value $t(i_j)\in \bbQ^+$ and a knapsack which can hold items with total weight at most $W\in \bbZ^+$. The objective is to pack a subset of the items into the knapsack so that their total value is maximized and their total weight is at most $W\in \bbZ^+$. The following is a well-known result about the Knapsack problem.

\begin{theorem}
\label{knapsack}
(\cite{IK75}) There is an algorithm for the Knapsack problem which finds the packing of items with maximum total value such that their total weight is at most $W$. The algorithm running time is $O(W \cdot q)$.
\end{theorem}
Note that the objective for sub-problem $DP(v,B_v, \hbp_v)$ is to find a tree ($H_v$), a subtree of $sub(v)$, and a subset of clients ($C'_v$) such that $\sum_{c \in C'_v} \theta_c$ is maximized with constraints that the weight of $H_v$ is at most $B_v$ and $\sum_{c \in C'_v} \theta_c \cdot \frac{d_T(c, H_v)}{t_c} \leq \hbp_v$. Consider an optimum tree $H^*_v$, there are only four different cases regarding if $H^*v$ contains children of $v$ or not. More formally, $H^*_v$ lies in one of the following cases: (Case 1) if $H^*_v$ contains none of the children of $v$, (Case 2) if $v$ has only one children and $H^*_v$ contains it, (Case 3) if $v$ has two children and $H^*_v$ contains only one of them, and finally (Case 4) if $v$ has two children and $H^*_v$ contains both of them. In the following we show how to find an optimum tree for each case. We select the one which collects the most profit as the solution of our algorithm. Because $H^*_v$ lies in one of the cases our solution is an optimum tree and the claim follows.

{\bf Case 1:} If $H^*_v$ does not contain any children of $v$ then resulting tree $H_v$ also contain none of the children of $v$ and hence $H_v$ contains only $v$. Therefore for any client $c \in sub(v)$ we have $d_T(c, H_v) = d_T(c, v)$ and hence the service cost of $c$ is $\theta_c \cdot \frac{d_T(c, v)}{t_c}$. Thus, the problem is only finding a subset of clients with maximum profit such that their total service cost is at most $\hbp_v$. In order to solve this problem we define the corresponding instance of the knapsack problem as follows. For each client $c \in C_v$ we create an item $c$ with value $t(c) = \theta_c$ and weight $w(c) = \theta_c \cdot d_T(c, v)$ in the corresponding knapsack instance. The weight bound $W$ is equal to $\hbp_v$. We use \cref{knapsack} to find a subset of items (clients) with maximum total profit where their total weight (service cost) is at most $\hbp_v$. The running time of the our algorithm is $O(\hbp_v \cdot |C|)$ as there are at most $|C|$ nodes in $sub(v)$.

{\bf Case 2:} Let $v_1$ be the only children of $v$. Because $H^*_v$ contains $v_1$, the resulting tree $H_v$ must contain $v_1$ and the length bound $B_v$ has to be at least $d_T(v, v_1)$, otherwise in this case the resulting tree is an empty tree with zero profit. Therefore an optimum solution for this case is the tree resulting from adding $v$ and its connecting edge $e(v, v_1)$ to an optimum solution of $DP(sub(v_1), B_v - d_T(v, v_1), \hbp_v)$. The set of served clients by the solution of $DP(sub(v_1), B_v - d_T(v, v_1), \hbp_v)$ union all the clients reside in $v$. Note that as $DP(sub(v_1), B_v - d_T(v, v_1), \hbp_v)$ is a smaller sub-problem we can assume that we have an optimum solution for this sub-problem. The running time for this case is $O(1)$.

{\bf Case 3:} Let $v_1, v_2$ be the two children of $v$. We explain for the case when $H^*_v$ contains $v_1$ but not $v_2$, the alternative is symmetric. We build $H_v$ as follows. Similar to Case 2, $B_v$ has to be at lease $d_T(v, v_1)$ in order to compensate for connecting $v_1$ to $v$. In this case we have the total budget of $\hbp_v$ for the service costs of both the clients in $sub(v_1)$ and clients in $sub(v_2)$. If we spend $\hbp_{v_1}$ for the service costs of the clients in $sub(v_1)$ then the budget for the service costs of clients in $sub(v_2)$ is $\hbp_{v_2} = \hbp_v - \hbp_{v_1}$. Note that the service cost for the clients reside in $v$ is zero therefore we serve them for free. Given $\hbp_{v_1}$ we can find an optimum tree and hence the optimum set of served clients in $sub(v_1)$ by taking the solution of the smaller sub-problem $DP(sub(v_1), B_v - d_T(v, v_1), \hbp_{v_1}$. Note that the nearest node in $H_v$ to the clients in $sub(v_2)$ is actually $v$ as we know $H_v$ does not contain $v_2$. Therefore given $\hbp_{v_2}$, we solve the corresponding instance of the Knapsack problem for the nodes in $sub(v_2)$ with weight bound $\hbp_{v_2}$ similar to Case 1. We iterate over all possible values of $\hbp_{v_1}$ from the set $\{0, 1, \ldots, \hbp_v\}$ and take the value for which we can collect the most profit from the clients in $sub(v_1)$ and clients in $sub(v_2)$. The running time for this case is $O(\hbp^2_v \cdot |C|)$ because we a loop of $O(\hbp_v)$ to iterate over the values of $\hbp_{v_1}$ where for each iteration we have to solve an instance of the Knapsack problem with $O(\hbp_v \cdot |C|)$.

{\bf Case 4:} Let $v_1, v_2$ be the two children of $v$. Similar to Case 2, $B_v$ has to be at lease $d_T(v, v_1) + d_T(v, v_2)$ in order to compensate for connecting both $v_1$ and $v_2$ to $v$ in the resulting tree $H_v$. Here, we have total service cost bound $\hbp_v$ and length bound $B_v - d_T(v, v_1) + d_T(v, v_2)$ to spend in the nodes in $sub(v_1)$ and $sub(v_2)$. If we spend $\hbp_{v_1}$ for the service costs for the clients in $sub(v_1)$ and $B_{v_1}$ for the length bound in $sub(v_1)$ we can spend at most $\hbp_v - \hbp_{v_1}$ for the service costs and $B_v - \hbp_{v_1} - d_T(v, v_1) - d_T(v, v_2)$ for the length in $sub(v_2)$. Given $\hbp_{v_1}$ and $B_{v_1}$ we can find an optimum tree and hence the optimum set of served clients in $sub(v_1)$ by taking the solution of smaller sub-problem $DP(sub(v_1), B_{v_1}, \hbp_{v_1})$ and find an optimum tree and hence the optimum set of served clients in $sub(v_2)$ by taking the solution of smaller sub-problem $DP(sub(v_2), B_v - B_{v_1} - d_T(v, v_1) - d_T(v, v_2), \hbp_v - \hbp_{v_1})$. We iterate over all possible values of $B_{v_1}$ from the set $\{0, \ldots, B_v - d_T(v, v_1) - d_T(v, v_2)\}$ and all possible values of $\hbp_{v_1}$ from the set $\{0, \ldots, \hbp_v \}$ and take the combination which collects the maximum profit from the clients in $sub(v_1)$ and $sub(v_2)$. The running time for this case is $O(B_v \cdot \hbp_v)$ as for each combination of values for $B_{v_1}$ and $\hbp_{v_1}$ we need $O(1)$ operations and there are at most $O(B_v \cdot \hbp_v)$ combinations.

Note that the running time of our algorithm is the sum of running time of all the cases and hence $O(B_v \cdot \hbp_v + \hbp^2_v \cdot |C|)$.
\qed
\end{proof}
Note that if node $u$ is a leaf then sub-problem $DP(u, B_u, \hbp_u)$ does not depend on any smaller sub-problem and can be solved efficiently by Claim \ref{dp-sub}. In order to avoid the dependency on the smaller sub-problems in Claim \ref{dp-sub} we solve the smaller sub-problems first. More formally, we run the following algorithm.
\begin{enumerate}
\item For each node $u\in V$ in the non-increasing order of the level number do:
\begin{enumerate}
\item For each value $B_u$ in $\{0, \ldots, B\}$ do:
\begin{enumerate}
\item For each value $\hbp_u$ in $\{0, \ldots, \hbp\}$ do:
\begin{enumerate}
\item Solve sub-problem $DP(sub(u), B_u, \hbp_u)$ using Claim \ref{dp-sub} and store the solution in the memory.
\end{enumerate}
\end{enumerate}
\end{enumerate}
\end{enumerate}
Note that in the above algorithm when we use Claim \ref{dp-sub} to solve sub-problem $DP(sub(u), B_u, \hbp_u)$ all its sub-problems are solved beforehand. The solution to the STSCST problem is the solution for the sub-problem $DP(r, B, \hbp)$ by its definition (see \cref{def-sub}). As we have the total of $O(n \cdot B \cdot \hbp)$ sub-problem and each of them take $O(B_v \cdot \hbp_v + \hbp^2_v \cdot |C|))$ time to solve by Claim \ref{dp-sub}, the total running time of our algorithm is $O(n\cdot B^2 \cdot \hbp^2 + n \cdot |C| \cdot B \cdot \hbp^3)$.
\qed
\end{proof}

In the following we show how \cref{dp-gen} helps us solve the NPCST problem and finishes the prove of \ref{npcst-thm}.
Note that the maximum possible value of $B'$ is $\sum_{c \in C} \frac{M}{\min_c t_c} \cdot \theta_v$ where $M$ is the distance between the farthest pair of nodes over all the trees in the set $\{T_1, T_2, \ldots, T_p\}$.
In order to solve instance $\I = (V, d, r, C, L)$ of the NPCST problem, for each tree $T_i$ in distribution $\pi$ and each value $2^j$ where $0 \leq j \leq \left\lceil \log \left(\sum_{c \in C} \frac{M}{\min_c t_c} \cdot \theta_v\right) \right\rceil$, we define  a corresponding instance $\I_{i,j} = (T_i, r, C, 4 \cdot A \cdot \log n \cdot L , 2^j)$ of TSCST as follows. The nodes set, root, and clients in $\I_{i,j}$ are the same as $\I$ but instead of distance function $d$ we have the distance function $d_{T_i}$.  Let $H_{i,j}$ be the solution for instance $\I_{i,j}$ using \cref{dp-gen}. We transform tree $H_{i,j}$ to the original metric ($\M(V, d)$) to obtain tree $\hat{H}_{i,j}$. Lets $C_{i,j}$ be the set of all served clients by $\hat{H}_{i,j}$ in $d$ with violation $16 \cdot A \cdot \log n $ in the neighborhoods, \ie $C_{i,j}$ contains each client $c \in C$ such that the neighborhood ball $\B(c, 16 \cdot A \cdot \log n \cdot t_c)$ is hit by $\hat{H}_{i,j}$. Our solution to the NPCST problem is tree $\hat{H}_{i,j}$ for $i \in [p]$ and $j \in [A \log n]$ whose $\sum_{c \in C_{i,j}} \theta_c$ is maximized.

In the following we show that $\hat{H}_{i,j}$ indeed serves $\frac{1}{2} \opt$. Let tree $T^*$ be an optimum solution to the NPCST problem in the original metric $d$ which has maximum length $L$ and it hits $\B$-ball of a subset ($C^* \subseteq C$) of the clients such that their total profit is optimum ($\opt = \sum_{v \in C^*} \theta_v$). Similar to \cref{def-tscst}, we define the service cost of $T^*$ for the clients in $C^*$ to be $\sum_{c \in C^*} \theta_v \cdot \frac{d(v, T^*)}{t_c}$. Note that as $T^*$ touches the neighborhood ball ($\B(c,t_c)$) of each client $c$ that it serves, $\sum_{c \in C^*} \theta_v \cdot \frac{d(v, T^*)}{t_c}$ can be at most $\opt$ (Fact 1).

The expected service cost of $T^*$ in the distribution $\pi$ of the trees is at most $A \cdot \log n \cdot \opt$ by the following equations.
\begin{align*}
\sum_{i = 1}^p \pi(i) \sum_{c \in C^*}  \frac{d_{T_i}(T^*,c)}{t_c} \cdot \theta_c &= \sum_{c \in C^*} \sum_{i = 1}^p \pi(i) \frac{d_{T_i}(T^*,c)}{t_c} \cdot \theta_c & \text{Linearity of the expectations}\\
&< A \cdot \log n \cdot \sum_{c \in C^*} \frac{d(T^*,c)}{t_c} \theta_c & \text{By Inequality (\ref{frt})}\\
&= A \cdot \log n \cdot \opt
\end{align*}
Similar to the above we can conclude that the expected length of $T^*$ in distribution $\pi$ of the trees is $A \cdot \log n \cdot L$.
By Markov's inequality, we can conclude that there are at least $\frac{3}{4} \cdot p$ trees in distribution $\pi$ in which length of $T^*$ is at most $4 \cdot A \cdot \log n \cdot L$. Similarly, there are at least $\frac{3}{4}\cdot p$ (possibly different) trees in distribution $\pi$ in which the service cost of $T^*$ is at most $4 \cdot A \cdot \log n \cdot \opt$.  The intersection of these two sets of trees contains at least $\frac{1}{2} p$ trees.  Therefore, among $T_1, T_2, \ldots, T_p$ there exists at least one tree $T_k$ whose corresponding instance $(T_k = (V,E), r, C, 4 \cdot A \cdot \log n \cdot L, 4 \cdot A \cdot \log n \cdot \opt)$ of TSCST (see \cref{def-tscst}) has a solution which collects at least $\opt$ profits. Hence, the optimal value of instance $\I_{k, j} = \left\lceil \log(4 \cdot A \cdot \log n \cdot \opt) \right\rceil$ of TSCST that we run in our algorithm is at least $\opt$ (Fact 2).

First note that the total number of instances $\I_{i,j}$s, we solve for TSCST is at most polynomial in terms of input as $i \in O(n \log n)$ and $0 \leq j \leq \left\lceil \log \left(\sum_{c \in C} \frac{M}{\min_c t_c} \cdot \theta_v\right) \right\rceil$. Consider our algorithm when it solves instance $\I_{i,j}$ where $i = k$ and $j = \left\lceil \log (4 \cdot A \cdot \log n \cdot \opt) \right\rceil$ which are taken from Fact 2. We prove that the total profit of the clients in $C_{i,j}$ served by tree $\hat{H}_{i,j}$ is at least $\frac{1}{2} \opt$ which completes the proof of \cref{npcst-thm}.
From Fact 2 we know that $\hat{H}_{i,j}$ collects $\opt$ profits with total service cost at most $8 \cdot A \cdot \log n \cdot \opt$. Now consider set $C'_{i,j}$ of all the clients $c$ whose $\frac{d\left(c, \hat{H}_{i,j}\right)}{t_c}$ is at most $16 \cdot A \log n$, from Markov's inequality we can conclude $\sum_{c \in C'_{i,j}} \theta_c$ is at least $\frac{1}{2} \sum_{c \in C_{i,j}} \theta_v = \frac{1}{2} \opt$. Therefore, $\hat{H}_{i,j}$ whose length is at most $4 \cdot A \cdot \log n \cdot L$ collects at least $\frac{\opt}{2}$ profits while violates the neighborhoods by factor $16 \cdot A \log n$.

\subsection{Description of an $\left(O(1), O(1), O(\log n)\right)$-Approximation Algorithm for the NPCST Problem}
\label{ref-alg}
Let $B_c$ denotes $\B(c,t_c)$. We ignore clients $c$ for which  $d(r,B_c)>L$ (these nodes are not covered in any optimal solution). Let $\theta_{max}$ denote the maximum profit of the clients among the remaining clients. Consider the following natural LP-relaxation. We have edge variables $\{x_e\}$ indicating if $e$ is in the resulting tree or not, and a variable $z_c$, $0\leq z_c\leq 1$, for every client $c$ indicating if $B_c$ is touched by the tree. The objective function of the LP is to maximize  $\sum_{c\in C} \theta_c\cdot z_c$, and the constraints are as follow. The first set of constraints are $\sum_{e \in \delta(S)} x_e\geq z_c$, for every node $c$ and set $S \subseteq V$ containing $B_c$ and not containing $r$ where $\delta(S)$ is the set of all the edges that have exactly one endpoint in $S$. The second set of constraints are $\sum_e d_e \cdot x_e\leq L$ ensuring the cost constraint of the tree. 

This LP can be solved efficiently by noting that there is a separation oracle for the first set of constraints in which we contract each set $B_c$ and use min-cut max-flow theorem between the contracted super node and root $r$. Let $(x^*,z^*)$ be an optimal solution and $K = \sum_c z_c$ . We can ignore all nodes $c$'s such that $z_c\leq 1/K^2$; this decreases the objective value by at most $\theta_{max}/K$ which is less than or equal to $\opt/K$. The remaining $z_c$'s can be bucketed into $O(log K)$ groups, where the $z_c$ values in each group are within a factor of $2$ of each other. The contribution from one of these buckets is at least $\OPT/\log K$, and we focus on such a bucket $C'\subseteq C$. Now perform a facility-location style clustering of the $B_c$ balls for the clients in $C'$ (this part is essentially the same as what we did for \cref{max-mr}). repeatedly pick the client with smallest neighborhood radius $t_c$, and remove all terminals $w$ such that $d(c,w)\leq 10\max(t_c,t_w)=10t_w$ (the constant $10$ is somewhat arbitrary), and set $nbr(w)=c$. Now, we can pretend that $w$ is co-located with $c$, and that $z_w=z_c$ and hence the first set of constraints still holds; this loses only a constant factor in the profit (since $z_c$ and $z_w$ are within a factor of $2$ of each other) and gives a constant-factor violation in the radius. We show the subset of $C'$ that we pick as the cluster centers by $C''$. Now, we contract the $B_c$'s for the cluster centers, and set the profit of the contracted node to be $\sum_{w\in C':nbr(w)=c\ or\ w=c}\theta_w$. 

Note that $(x^*,1-z^*)$ induces a solution to the Prize Collecting Steiner Tree (PCST) problem  instance defined by the contracted nodes. We can build a tree $T$ of cost $O(L)$ such that the expected profit of the contracted nodes connected by $T$ is in fact at least $\sum_{c \in C''}\theta_c \cdot z_c$ which is in $\Omega(\opt/\log k)$. This can be done by using results of Bang-Jensen et al. \cite{BFJ95}, which gives another way of obtaining a $2$-approximation algorithm for Prize Collecting Steiner Tree (PCST). The result of Bang-Jensen et al. \cite{BFJ95} implies that from $x^*$ we can build a convex combination of trees such that each client $c$ appears in $z^*(c)$ fractional number of the trees.  Note that $\sum_{c \in C''}\theta_c \cdot z_c = \Omega(\opt/\log k)$ and the average cost of the trees is $\sum_e x^*_e \leq L$. By Markov inequalities we can conclude that at least $\frac{3}{4}$ of the trees have cost at most $4\cdot L$. Similarly, by Markov Inequalities at least $\frac{3}{4}$ of the trees collect at least $\frac{1}{4} \Omega(\opt/\log k)$ amount of profit. Therefore we conclude that at least half of the trees have length at most $4 \cdot L$ and collect at least $\frac{1}{4} \Omega(\opt/\log k)$ amount of profit. Let $T$ be one of them.

Now we take $T$ and uncontract the super nodes. We need to give a connected graph, which can be done  by connecting all the nodes of $T$ incident to $B_c$ to $c$, where $c$ is a cluster center. Since any two cluster centers are far apart, the extra cost incurred in this can be charged to the edges of $T$ incident to the contracted node corresponding to $B_c$.

%% file: euclidean.tex
\section{Proof of \cref{npcst-2d}}
\label{sec:euclidean}
We prove formally \cref{npcst-2d} in the $2$-dimensional Euclidean space. We believe that our method can be generalized to the $D$-dimensional space; the changes needed to do so are brought at the end of this section. We ignore optimizing the constants to keep the algorithm and its analyze easier to explain.

Let us assume that the radius of the biggest neighborhood is $\frac{M}{2}$. We tile the plane with regular hexagons of side length $M$ so that every point in the plane is covered with exactly one tile and root node $r$ is at the center of a hexagon (as shown in \cref{clarfig}).

\begin{figure}[!h]
  \centering
      \includegraphics[scale=0.7]{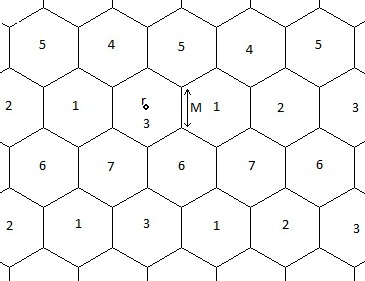}
  \caption{Tiling the plane with regular hexagons. The digits in each hexagon denotes its color. Root node $r$ is centered at one of the tiles and the length of the edges of the hexagons is $M$.}
  \label{clarfig}
\end{figure}

We color the hexagons with $7$ colors such that no two neighboring hexagons get the same color. This coloring can be done by coloring each non-colored hexagon with one of the $7$ colors that is not used in a neighboring colored hexagon. The coloring is possible for all the non-colored hexagon since it has $6$ neighboring hexagon and there are $7$ colors  (see \cref{clarfig}).

We construct an auxiliary graph from $G$ in order to assign the profit of each neighborhoods to a single node which relaxes the problem to the Budgeted Prize Collecting Steiner Tree (BPCST) \cite{JMP00}. In the BPCST problem we are given graph where all the nodes have a profit, a root node $r$, and a budget $L$; the objective is to find a tree with length at most $L$ which the sum of profits of the nodes it contain is maximum. There is a $(4 + \epsilon)$-approximation algorithm for BPCST due to Chekuri \etal \cite{CKP08}.

We call a hexagonal tile an {\em occupied-tile} if it contains at lest one node of $G$. We correspond each occupied-tile to a node in $V$. We call the nodes that are corresponded to the occupied tiles {\em center} nodes. If the occupied-tile contain root $r$ then its center node is $r$ otherwise the center node is an arbitrary node which is inside the tile. We construct the auxiliary graph ($\hvG = (\hvV, \hvE)$) as follow. $\hvV \subseteq V$ is the set of all the center nodes that correspond to occupied-tiles so that there is a bijection between nodes in $\hvV$ and the occupied-tiles. Note that root $r$ is a center node and is in $\hvV$ since $r$ is assigned to an occupied-tile. $\hvE \subseteq E$ is the subset of the edges in $E$ that have both endpoints in $\hvV$. We assign each client to all the occupied tiles its $\B$-ball intersects since each client resides in a node of $G$ at least one occupied-tile exists for each client. The profit of each center node in $\hvV$ is the sum of profits of the clients assigned to its corresponding occupied-tile.

Neighborhood Prize Collecting Steiner Tree Algorithm (NPCSTA) showed in \cref{2d-alg} is our $(4P + 1, 35, 12 + \epsilon)$-approximation algorithm for Euclidean NPCST where $P$ is the ratio of the largest neighborhood radius over the shortest one. Remember NPCSTA being $(4P+1, 35, 12 + \epsilon)$-approximation algorithm for the instance $(V, d, r, C, L)$ of the ENPCST problem, means it finds a Steiner tree $T$ with the following properties. $T$ is said to hit a client $c$ with neighborhood ball $\B(c,t_c)$ if $T$ has at least one node in $\B(c, t_c \cdot (4P + 1))$, the weight of $T$ is at most $(35 \cdot L$, and sum of the profits of the clients got hit by $T$ is at least $\frac{1}{12 + \epsilon} \opt$ where $\opt$ is the amount of profit an optimum tree collects with no violation in any bound. We use the $(4 + \epsilon')$-approximation algorithm of \cite{CKP08} for the BPCST problem as a black box in our algorithm (here we set $\epsilon' = \frac{\epsilon}{3}$).

\begin{figure}[!h]
\noindent
\fbox{\parbox{\textwidth}{
\begin{enumerate}
\item For each color $i =1, 2, \ldots, 7$:
\begin{enumerate}
\item Find a tree $\hvT_i$ rooted at $r$ in $\hvG$ with weight at most $5 \cdot L$ which collects the maximum profit only from the nodes that correspond to the occupied-tiles of color $i$ using the $(4+\epsilon')$-approximation algorithm of \cite{CKP08} for the BPCST problem.
\end{enumerate}
\item Return $T = \hvT_1 \cup \ldots \cup \hvT_7$ as the result. The set of covered clients are all the clients that are assigned to the occupied-tiles whose center node is in $T$.
\end{enumerate}
}}
\caption{Neighborhood Prize Collecting Steiner Tree Algorithm}
\label{2d-alg}
\end{figure}

We prove the following lemma which implies \cref{npcst-2d}.
\begin{lemma}
NPCSTA (showed in \cref{2d-alg}) is a $(4P +1, 35, 12 + \epsilon)$-approximation algorithm for the ENPCST problem where $P$ is the ratio of the largest neighborhood radius over the shortest one.
\end{lemma}
\begin{proof}
We use the structure of an optimum solution in order to show that there is also a tree in the graph $\hvG$ which also collects the same profit as the optimal tree with small increases in the neighborhood balls and the weight bound.
Let $T^*$ be an optimal tree collecting $\opt$ profit. Let $\hvT^*_i$ be the projection of $T^*$ on the set of nodes in $\hvG$ that correspond to the tiles of color $i$ plus root node $r$. More specifically, we build $\hvT^*_i$ as follows. If $T^*$ contains a node in a tile of color $i$ then $\hvT^*_i$ contains the node in $\hvG$ that corresponds to the tile of color $i$. $\hvT^*_i$ also contain root node $r$. There is an edge ($\hve = \hvv_1 \hvv_2$) between two nodes ($\hvv_1$, $\hvv_2$) of $\hvT^*_i$ if and only if there is a node $v_1$ in $T^*$ that resides in the tile corresponding to $\hvv_1$ and a node $v_2$ in $T^*$ that resides in the tile corresponding to $\hvv_2$ such that the unique path between $v_1$ and $v_2$ in $T^*$ does not pass through any node that resides in a tile that corresponds to a node of $\hvT^*_i$. The call the unique path which was the reason for adding an edge between $\hvv_1$ and $\hvv_2$ a {\em creator} path.  In the following claim we prove that the weight of $\hvT^*_i$ is at most $5$ times larger than $T^*i$.

\begin{figure}
        \centering
        \fbox{
        \begin{subfigure}[b]{0.4\textwidth}

                \centering
                \includegraphics[width=\textwidth]{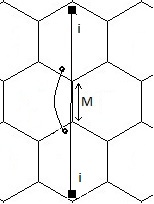}
                \caption{An edge between two center nodes in the tiles of color $i$ compared to its corresponding path in $T^*$ (creator path) shown as the curved line.}
        \label{itoi}
        \end{subfigure}%
        }\quad
        \fbox{
        \begin{subfigure}[b]{0.4\textwidth}
                \centering
                \includegraphics[width=\textwidth]{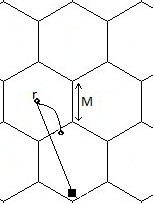}
                \caption{An edge between the root node ($r$) and a center node in $\hvT^*_i$ compared to its corresponding path in $T^*$ (creator path) shown as the curved line.}
                \label{rtoi}
        \end{subfigure}
        }
        \caption{Edges in $\hvT^*_i$ compared to paths in $T^*$. The squares represents the center nodes in $\hvV$ and circles represent the non-center nodes in $V$. The straight lines are the edges of $\hvT^*_i$ and the curvy lines are the paths of $T^*$ (the creator paths for the edges of $\hvT^*_i$).}
        \label{fig:animals}
\end{figure}
\begin{nclaim}
\label{iprojection}
The weight of $\hvT^*_i$ is at most $5$ times larger than the weight of $T^*_i$.
\end{nclaim}
\begin{proof}
The weight of $\hvT^*_i$ can be larger than $T^*_i$ as we select the center nodes of the tiles instead of the original nodes in $T^*_i$. By the construction $\hvT^*_i$, each of its edge is added to $\hvT^*_i$ because of a path in $T^*_i$. Moreover each path of $T^*_i$ contribute in adding at most one edge in $\hvT^*_i$ since only the endpoints of the path are in the occupied-tiles whose center is in $\hvT^*_i$ and all the other nodes are inside the other occupied-tiles (neither colored $i$ nor contain $r$). Therefore, we can charge each edge of $\hvT^*_i$ to their corresponding path. There are two types of edges in $\hvT^*_i$. The first type (Type One) contains the edges between two center nodes of tiles with color $i$ (see \cref{itoi}) and the second type (Type Two) contains the edges between $r$ and a center node of a tile with color $i$ (see \cref{rtoi}).

Note that we require no two tiles of color $i$ be neighbor of each other (This paragraph indeed shows the crucial point of coloring). Therefore, each creator path in $T^*i$ has length at least $M$ if its endpoints are the closest points in the two tiles (see Figure \cref{itoi}). The corresponding edge in $T^*_i$ can be at most $4M$ units of length larger than its creator path since its two center nodes can be the farthest points of the two tiles as opposed to the creator path (see \cref{itoi}). Therefore in the worst case the length of the Type One edges in $\hvT^*_i$ are $5$ times larger than its creator path.

The end points of creator paths for the Type Two edges can be in the neighboring tiles. The length of the creator paths for the Type Two edges is at least $\frac{\sqrt{3}}{2}M$ which is the distance of the center of the tile (node $r$) to the middle of a side (see \cref{rtoi}). The length of the corresponding edge in $\hvT^*_i$ can be at most $2M$ units of length larger than the creator path since the center node of the tile with color $i$ can be at the farthest point of the tile (Not that the distance between two farthest point in a tile is $2M$). Therefore in the worst case the maximum multiplicative inflation in the Type Two edges of $\hvT^*_i$ happens when the tiles are neighbors and the multiplicative factor is $\frac{(2 + \sqrt{3}/2)M}{\sqrt{3}/2\cdot M} \thickapprox 3.31$ (see \cref{rtoi}, in fact the multiplicative factor is smaller but we upper bound it by $3.31$).

Because the multiplicative increase compared to the creator paths for the Type One edges is at most $5$ and for the Type Two edges is at most $3.31$, the weight of $\hvT^*_i$ is at most $5$ times larger than the weight of $T^*_i$.
\end{proof}

In order to prove lemma we show that the resulting tree of NPCSTA ($T$) satisfies all the three criteria in the following three claims.
We start with the first criteria to show that the multiplicative increase in the neighborhoods' radius is $4P + 1$.

\begin{nclaim}
For each client $c$ that is in the set of served clients by NPCSTA, $T$ visits a node whose distance from $c$ is $4P+1$ times the radius of the neighborhood of $c$ ($t_c$).
\end{nclaim}
\begin{proof}
Let assume that $c$ is in the set of served clients because its neighborhood intersects with a tile of color $i$ whose center is in $T_i$. Similar to proof of Claim \ref{iprojection} the center node of the tile can be far from the point where neighborhood of $c$ intersects. Note that this distance is at most $2M$. Therefore $T_i$ visits a node which is at most $2M + t_c$ away from $c$ where $t_c$ is the radius of $c$'s neighborhood. Because the ratio of the largest neighborhood to the smallest one is $P$, $T_i$ and hence $T$ contain a node which has distance $(4P+1) \cdot t_c$ from $c$.
\qed
\end{proof}

\end{proof}
In the following we prove that the weight of the output of NPCSTA ($T$) is at most $35L$ (the second criterion of the tri-criteria approximation algorithm).
\begin{nclaim}
Weight of $T$ is at most $35L$.
\end{nclaim}
\begin{proof}
Weight of each tree $T_i$ where $i \in \{1, \ldots, 7\}$ is $5L$ and all of them contain root $r$ and hence connected to each other. Weight of $T$ is at most $35L$ because it is the union of all the $T_i$s.
\qed
\end{proof}
The following claim proves the bound of the third criterion of NPCSTA and finishes the proof of the lemma.
\begin{nclaim}
$T$ collects at least $\frac{1}{12 + \epsilon}$ fraction of the optimal profit.
\end{nclaim}
\begin{proof}
Let $\opt_i$ be the amount of profit $T^*$  collects from the clients ($C_i$) that are assigned to a tile of color $i$. Because the side length of each hexagon is $M$ and no two tiles with color $i$ intersects; the distance between two tiles of color $i$ is at least $M$. As the radius of the largest neighborhood is $\frac{M}{2}$ the $\B$-ball of each client in $C_i$ is intersecting to exactly one tile of color $i$. Therefore, there is a unique way of assigning clients in $C_i$ to the tiles of color $i$. From Claim \ref{iprojection} we can conclude that there is a tree $\hvT^*_i$ which visits all the center nodes of the tiles with color $i$ that contain at least one node of $T^*$. Therefore $\hvT^*$ collects at least $\opt_i$  profit from the color $i$ tiles and has weight $5L$. Thus, $T_i$ which is the result of $(4+\epsilon')$-approximation algorithm of \cite{CKP08} collects at least $\frac{1}{4 + \epsilon'} \opt_i$.

Note that the radius of the maximum neighborhood is $\frac{M}{2}$ and the side length of the hexagons is $M$. Each neighborhood can intersects with at most $3$ hexagons since if it intersects with $4$ tiles two of them cannot be neighbor of each other and have distance $M$. Therefore the profit of each client contribute in at most $3$ colors. From each color $i$, tree $T$ collects at least $\frac{1}{4 + \epsilon'} \opt_i$ and each client's profit can appear in at most $3$ colors; the proof of the Claim follows by summing over all the $7$ colors and setting $\epsilon' = \frac{\epsilon}{3}$.
\qed
\end{proof}
The approach for generalizing the above algorithm to the $D$-dimensional space is the same as NPCSTA (\cref{2d-alg}). We decompose the space into tiles using $D$-dimensional cubes and proceed as before. Each cube touches at most $O(2^D)$ other cubes so we have at most $O(2^D)$ colors. The farthest points between two points in a $D$-dimensional cube is $O(2^D)$. Therefore, we believe that our algorithm can be generalized to the $D$-dimensional Euclidean space and get $(O(2^D), O(2^D), O(2^D))$-approximation algorithm. 

%% file: maxmr.tex
\section{Max-MR problem}
\label{max-mr}
The other objective which is usually considered in the movement framework is to minimize the maximum latency. This objective is taken from the applications when there is a deadline by which all the clients have to be served. In our application the deadline specifies the latest time we can serve the last client which can be referred to as minimizing the maximum latency. We refer to this objective in our setting as the Max-MR problem. We give a constant-factor approximation algorithm for the case when all repairmen have the same speed and prove \cref{thm:max-mr}. Client serving problems with max objective are studied thoroughly for lots of different scenarios \cite{FHK76,li,EGK+04,arkin,campbell,xuxu,XW,KS11,nagarajan}.

\begin{proof}
We start by guessing the optimum maximum latency ($\opt$) by which all the clients will be served. Our algorithm for a given guessed value ($T$) for the maximum latency, either finds a path for each repairman that serve all the clients with latency at most $10T$ or announce that $T \leq \opt$. Therefore, we can find an appropriate value $T$ such that $\opt \leq T \leq \opt (1 + \epsilon)$ for a small positive constant $\epsilon$ by binary search over the range $[1, \frac{2 \cdot MST(G)}{\min_i{v_i}}]$. The upper value in the range is indeed an upper bound for the maximum latency which is the time required for the slowest repairman to visit all the nodes in the graph with a path obtained by doubling the edges of a MST.

Our algorithm proceed as follows. For a given $T$ we assign to each client $c$ a neighborhood $\B(c, t_c) = \{v | d(v, c) \leq t_c\}$ of radius $t_c = v'_c \cdot T$ where $v'_j$ is the speed of $c_j$. Now, for each client $c$ at least one repairman should visit at least one node in $\B(c_j, t_c)$. We run a clustering algorithm to cluster these neighborhoods as follows. We start with a client ($c$) whose neighborhood has the smallest radius and tag it as a {\em leader} client. We assign all the clients $c'$ for which $\B(c, t_c) \bigcap \B(c', 9t_{c'} \neq \emptyset$ to the client $c$ and tag them {\em slave} clients, \ie a client is slave if its stretched neighborhood with radius $9$ times the original radius intersects with the neighborhood of a leader client.  After that, we discard all the tagged clients and make a non-tagged client with the smallest neighborhood radius a leader and proceed as before. We continue tagging until all the clients get tagged.

\begin{nclaim}
\label{claim:extraedges}
The distance between the neighborhoods of any two leader clients ($c$ and $c'$) is at least $8$ times the radius of the larger neighborhood. Here the distance between neighborhood $\B(c, t_c)$ and neighborhood $\B(c', t_{c'})$ is the distance between the closest pair of nodes one in $\B(c, t_c)$ and the other in $\B(c', t_{c'})$.
\end{nclaim}
\begin{proof}
Without loss of generality assume $t_c \geq t_{c'}$. We prove the claim by contradiction. Assume the distance is less than $8t_c$, then $\B(c, 9 t_c)$ intersects with the neighborhood of $c'$. This makes $c$ tagged as a slave client of $c'$ in the process of tagging when we tagged $c'$ as a leader client.
\end{proof}

We contract each leader ball to a super-node maintaining the metric properties to obtain graph $\hvG$. We define an instance of the rooted version of the min-max $k$-tree cover \cite{EGK+04,arkin} over $\hvG$. In the min-max $k$-tree cover we are given a set of terminals in a metric space along with $k$ root nodes. The objective is to find $k$ trees rooted in the root nodes to cover all the terminals in the graph such that the length of the maximum tree is minimized. In $\hvG$ the terminals are the super-nodes and the root nodes are the starting location of the repairmen (which can be contracted into a super-node).

We use the $4$-approximation algorithm of Even et al.~\cite{EGK+04} for the min-max $k$-tree cover to find $k$-trees in $\hvG$ rooted at the root nodes and covering all the super-nodes. The set of paths for the repairmen are constructed as follows. For each repairmen $r$ we take the tree ($\hvT_r$) rooted at the starting location of $r$ found by the algorithm of Even et al.~\cite{EGK+04}. We double the edges of $\hvT_r$, build an Eulerian walk, and obtain path $\hvP_r$. Now, consider the edges of $\hvP_r$ in the graph $G$ when we uncontract each super-node to the neighborhood ball of its corresponding leader client. We reconnect $\hvP_r$ in $G$ by adding an extra edge for each pair of separated nodes that were a single super-node in $\hvP_r$. Moreover, if the root node is in a super-node and is disconnected from the rest of the path, we reconnect it by adding an extra edge to the node that was in the same super node as the root in $\hvP_r$. Let the connected path in $G$ obtained from $\hvP_r$ be $P_r$. If length of $P_r$ is less than $10\cdot v \cdot T$ for each repairman $r$ where $v$ is the common speed of all the repairmen, then set $\{P_r\}_{r \in R}$ is the result of our algorithm for the given $T$. Moreover, each client gets service from the closest node in one of the paths in $\{P_r\}_{r \in R}$. Otherwise if length of $P_r$ is greater than $10\cdot v \cdot T$ our algorithm announces that $T$ is less than $\opt$ (the optimum solution of Max-MR defined over $G$). We prove the following lemma which proves the correctness of our algorithm and finishes the proof of \cref{thm:max-mr}.

\begin{lemma}
If the guessed value $T$ is not smaller than $\opt$ then the length of $P_r \leq 10 \cdot v \cdot T$ for each repairman $r$ and the maximum latency of the clients is $10T$.
\end{lemma}
\begin{proof}
Let assume $\opt \leq T$. In the following claim we prove that the length of the paths of the repairmen is at most $10 \cdot v \cdot T$.
\begin{nclaim}
\label{maxrepairman}
Let repairman $r$ be the repairman whose $P_r$ has the maximum length. The length of $P_r$ is at most $10 \cdot v \cdot T$.
\end{nclaim}
\begin{proof}
Note that every path assigned to the repairmen in an optimal solution to the Max-MR has length at most $v \cdot \opt$. The optimal solution of the min-max $k$-tree cover problem defined over $\hvG$ is at most $v \cdot \opt$  because the paths in an optimal solution of Max-MR after contracted into the super nodes, are a candidate solution to the min-max $k$-tree cover in $\hvG$. Therefore the length of the maximum tree in the solution we obtain from the $4$-approximation algorithm is  at most $4 \cdot v \cdot \opt$ and hence it is at most $4\cdot v \cdot T$. Therefore the length of $\hvP_r$ is at most $8\cdot v \cdot T$ as it is obtained by doubling the edges of the tree assigned to $r$.

Path $P_r$ obtained from $\hvP_r$ by adding extra edges inside the neighborhood of the leader clients. We start from the starting location of $r$ and go along path $P_r$. Whenever we encounter an edge inside a leader client's neighborhood which is not in $\hvP_r$ we charge its length to the path in $\hvP_r$ that goes from the neighborhood of the current leader client to the neighborhood of the next leader client. From Claim \ref{claim:extraedges} we know that the length of the path to which we charged the length of the extra edge is at least $8$ times the radius of the neighborhood. Because the length of the extra edge is at most twice the radius of the neighborhood, the length of the extra edge is at most $\frac{1}{4}$ of the path it is charged to.  Therefore the total length of $P_r$ is at most $10\cdot v \cdot T$.
\qed
\end{proof}
In the following claim we prove that for each client $c$ there exist a node with distance at most $10 t_c$ from $c$ which is along the path of a repairman.
\begin{claim}
\label{maxclient}
For each client $c$ there exist at least one repairman $r$ such that there is at least one node which is in $\B(c, 10t_c)$ and in $P_r$.
\end{claim}
\begin{proof}
If $c$ is a leader client then the claim follows directly as the paths of the repairmen visit at least one node from the neighborhood ($\B(c, t_c)$ of each leader client. If $c$ is a slave client $\B(c, 9 t_c)$ intersects with the neighborhood of a leader client $c'$. Note that $t_c \geq t_{c'}$ since we tag the clients as leaders in the increasing order of their neighborhood radius. Therefore the distance of $c$ from any node in $\B(c', t_{c'})$ is at most $10 t_c$. As $c'$ is a leader client the paths of repairmen visit at least one node from $\B(c', t_{c'})$. Therefore, there exists a node which is in $\B(c, 10t_c)$ and resides in the path of a repairman.
\end{proof}
From Claim \ref{maxrepairman} we conclude that all the repairman can travel their assigned paths by time $10T$. From Claim \ref{maxclient} we conclude that each client by time $10T$ can go to a node which is in the path of a repairman. Therefore, the maximum latency of the clients is $10T$.
\end{proof}
\qed
\end{proof}